\documentclass{article}
\usepackage[utf8]{inputenc}
\usepackage{fullpage}
\usepackage{braket}
\usepackage[numbers]{natbib}
\usepackage{amsmath,amsfonts,amsthm,bm}

\def\eqref#1{equation~\ref{#1}}

\def\ceil#1{\lceil #1 \rceil}

\def\1{\bm{1}}

\def\vx{{\bm{x}}}

\def\mG{{\bm{G}}}
\def\mH{{\bm{H}}}
\def\mI{{\bm{I}}}

\def\mM{{\bm{M}}}

\def\mO{{\bm{O}}}

\def\mU{{\bm{U}}}
\def\mV{{\bm{V}}}

\def\mX{{\bm{X}}}
\def\mY{{\bm{Y}}}
\def\mZ{{\bm{Z}}}

\DeclareMathAlphabet{\mathsfit}{\encodingdefault}{\sfdefault}{m}{sl}
\SetMathAlphabet{\mathsfit}{bold}{\encodingdefault}{\sfdefault}{bx}{n}

\DeclareMathOperator{\Tr}{Tr}

\usepackage{amssymb}
\usepackage{mathtools}
\usepackage{amsthm}
\usepackage{amsmath}
\usepackage{thm-restate}

\usepackage{physics}
\usepackage{pifont}
\newcommand{\cmark}{\ding{51}}
\newcommand{\xmark}{\ding{55}}
\usepackage{upgreek}

\usepackage{booktabs}
\usepackage{hyperref}
\usepackage{cleveref}
\usepackage{authblk}

\usepackage{graphicx}
\usepackage{tikz}
\usetikzlibrary{quantikz}
\usepackage{caption}
\usepackage{subcaption}
\usepackage{multirow}

\usepackage[symbol]{footmisc}

\theoremstyle{plain}
\newtheorem{theorem}{Theorem}[section]

\newtheorem{corollary}[theorem]{Corollary}
\theoremstyle{definition}
\newtheorem{definition}[theorem]{Definition}

\theoremstyle{remark}

\newcommand{\ce}{\ensuremath{\mathrm{e}}}

\newcommand{\change}{}

\crefname{section}{Section}{Section}
\crefname{subsection}{Section}{Section}
\crefname{figure}{Figure}{Figure}

\title{Beyond Barren Plateaus: Quantum Variational Algorithms Are Swamped With Traps}
\author[1]{Eric R. Anschuetz\footnote{\href{mailto:eans@mit.edu}{\texttt{eans@mit.edu}}}}
\author[2]{Bobak T. Kiani\footnote{\href{mailto:bkiani@mit.edu}{\texttt{bkiani@mit.edu}}}}
\affil[1]{MIT Center for Theoretical Physics}
\affil[2]{Department of Electrical Engineering and Computer Science, MIT}

\date{}

\begin{document}

\maketitle

\begin{abstract}
    One of the most important properties of classical neural networks is how surprisingly trainable they are, though their training algorithms typically rely on optimizing complicated, nonconvex loss functions. Previous results have shown that unlike the case in classical neural networks, variational quantum models are often not trainable. The most studied phenomenon is the onset of barren plateaus in the training landscape of these quantum models, typically when the models are very deep. This focus on barren plateaus has made the phenomenon almost synonymous with the trainability of quantum models. Here, we show that barren plateaus are only a part of the story. We prove that a wide class of variational quantum models---which are shallow, and exhibit no barren plateaus---have only a superpolynomially small fraction of local minima within any constant energy from the global minimum, rendering these models untrainable if no good initial guess of the optimal parameters is known. We also study the trainability of variational quantum algorithms from a statistical query framework, and show that noisy optimization of a wide variety of quantum models is impossible with a sub-exponential number of queries. Finally, we numerically confirm our results on a variety of problem instances. Though we exclude a wide variety of quantum algorithms here, we give reason for optimism for certain classes of variational algorithms and discuss potential ways forward in showing the practical utility of such algorithms.
\end{abstract}

\section{Introduction}

The trainability of classical neural networks via simple gradient-based methods is one of the most important factors leading to their general success on a wide variety of problems. This is particularly exciting given the variety of no-go results via statistical learning theory, which demonstrate that in the worst case these models are \emph{not} trainable via stochastic gradient-based methods~\cite{blum1994weakly,szorenyi2009characterizing,goel2020statistical,shalev2017failures}. There has been recent hope that variational quantum algorithms---the quantum analogue of traditional neural networks---may inherit these nice trainability properties from classical neural networks. Indeed, in certain regimes~\cite{farhi2019}, training algorithms exist such that the resulting quantum model provably outperforms certain classical algorithms. This would potentially enable the use of quantum models to efficiently represent complex distributions which are provably inefficient to express using classical networks~\cite{gao2021enhancing}.

Unfortunately, such good training behavior is not always the case in quantum models. There have been previous untrainability results for deep variational quantum algorithms due to vanishing gradients~\cite{mcclean2018barren,cerezo2020costfunctiondependent,marrero2020entanglement,napp2022quantifying}, and for nonlocal models due to poor local minima~\cite{anschuetz2022critical}; however, no such results were known for shallow, local quantum models with local cost functions. Indeed, there have been promising preliminary numerical experiments on the performance of variational quantum algorithms in these regimes, but typically have relied on good initialization~\cite{cong2019} or highly symmetric problem settings~\cite{wiersema2020exploring,kim2020universal,kim2021quantum} to show convergence to a good approximation of the global optimum.

\begin{figure*}[t!]
 \captionsetup[subfigure]{aboveskip=-1pt,belowskip=-3pt}  
 \begin{subfigure}{0.5\textwidth}
    \includegraphics[trim={2.0cm 1.5cm 2.0cm 2.0cm},clip,width=\linewidth]{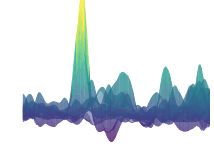}
    \label{fig:surface_3d}
  \end{subfigure}%
  \hspace*{\fill}   
 \begin{subfigure}{0.5\textwidth}
    \includegraphics[width=\linewidth]{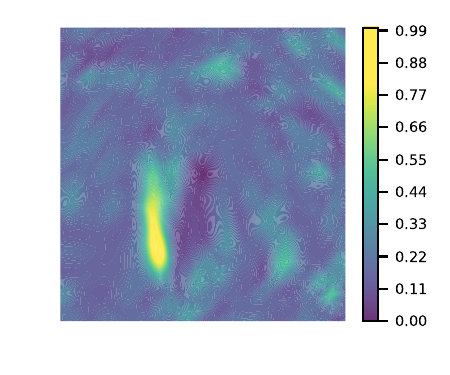}
    \label{fig:surface_contour}
  \end{subfigure}%
  \hspace*{\fill}   
\caption{Loss landscapes of underparameterized quantum variational algorithms generally appear ``bumpy,'' filled with various local minima and traps. Here, we plot the loss landscape as a surface (left) and contour (right) plot along two random normalized directions for the teacher-student learning task of the QCNN for 14 qubits (see \Cref{sec:teacher_student}). Though a global minimum is located at the center of the plot, finding this global minima is generally challenging due to the shape of the loss landscape. Details of this visualization are given in the Supplemental Information.} \label{fig:surface_both}
\end{figure*}

Here, we show that generally such models are \emph{not} trainable, particularly when a good choice of initial point is not known and when the model does not exhibit a high amount of symmetry. \change{We first prove general untrainability results in the presence of noise using techniques from statistical query learning theory. Surprisingly, these results hold for \emph{all} learning problems in a wide range of variational learning settings, and in many scenarios even when the magnitude of the noise is \emph{exponentially small} in the problem size.} We then \change{consider the trainability of models that may not have noise by studying their typical loss landscapes. We prove that, for typical model} instances, local minima concentrate far from the global optimum even for certain local shallow circuits that do not suffer from barren plateaus. This phenomenon can be visualized in Figure~\ref{fig:surface_both}, where the training landscape for a shallow QCNN learning a random instance of itself is shown to concentrate far from the global optimum. As in \cite{anschuetz2022critical}, this phenomenon is the result of a trainability phase transition in the loss landscape of the quantum model. In \cite{anschuetz2022critical}, this transition was governed by the ratio of the number of parameters to the Hilbert space dimension; we show in the shallow case that instead, this transition is governed by the ratio of the \emph{local} number of parameters to the \emph{local} Hilbert space dimension, in the reverse light cone of a given measured observable. As this is typically much less than $1$ for local variational ansatzes, these models are typically untrainable. We then give numerical evidence of this fact, and conclude by studying where there may be reason for optimism in the training of certain variational quantum models.

\section{Results}

\subsection{Preliminaries}

Quantum machine learning algorithms have been a focus of intense research effort as potential use-cases for noisy, intermediate-scale quantum (NISQ)~\cite{Preskill2018quantumcomputingin} devices. Just as in classical machine learning, algorithms are tasked with minimizing some risk:
\begin{equation}
    \mathcal{R}(f) = \mathbb{E}_{\vx}\left[\ell(f(\vx))\right],
\end{equation}
given a model $f$, a distribution of inputs $\bm{\vx}$, and a loss function $\ell$. To perform learning, one searches for a model $\widehat{f} \in \mathcal{F}$ in the function class $\mathcal{F}$ (e.g. the set of functions expressed by quantum neural networks). The expected risk $\mathcal{R}(\widehat{f})$ is typically not something one can calculate, as it requires access to the full probability distribution of the data. Instead, one often minimizes the \emph{empirical risk} $\widehat{\mathcal{R}}(\widehat{f})$ (often named the \emph{training error}) over a given training data set $D$:
\begin{equation}
    \widehat{\mathcal{R}}(\widehat{f}) = \sum_{\vx_i\in D}\ell(\widehat{f}(\vx_i)).
    \label{eq:empirical_risk}
\end{equation}

\change{Perhaps the most well-studied class of quantum machine learning algorithm consists of \emph{variational quantum algorithms} (VQAs)~\cite{peruzzo2014variational}. VQAs} are a class of quantum generative models where one expresses the solution of some problem as the smallest eigenvalue and corresponding eigenvector (typically called the \emph{ground state}) of an objective Hermitian matrix $\bm{H}$---called the \emph{Hamiltonian}---on $n$ qubits. Given a choice of generative model---often called an \emph{ansatz} in the quantum algorithms literature:
\begin{equation}
    \ket{\bm{\theta}}=\prod\limits_{i=1}^q\bm{U_i}\left(\theta_i\right)\ket{\psi_0}
    \label{eq:ansatz}
\end{equation}
that for some choice of $\bm{\theta}$ is the ground state of $\bm{H}$, the solution is encoded as the minimum of the loss function\change{
\begin{equation}
    \widehat{\mathcal{R}}_{\text{VQE}}(\bm{\theta})=\sum\limits_{i=1}^A \alpha_i\bra{\bm{\theta}}\bm{P}_i\ket{\bm{\theta}},
    \label{eq:loss_func}
\end{equation}
where:
\begin{equation}
    \bm{H}=\sum\limits_{i=1}^A\alpha_i\bm{P}_i
\end{equation}
is the Pauli decomposition of $\bm{H}$}. VQAs have found numerous applications~\cite{cerezo2021variational}, and a countless number of VQA instances have been proposed for various quantum learning tasks.

Typically, models in VQAs come in one of two flavors: Hamiltonian agnostic models, and Hamiltonian informed models. Hamiltonian agnostic models are constructed such that the $\bm{U_i}$ are independent of $\bm{H}$, and are generally chosen to be efficient to implement. This is most analogous to the case in classical generative modeling, where the model structure is usually independent from the specific choice of data $\bm{H}$. One might hope then that training Hamiltonian agnostic VQAs is completely analogous to the classical setting, then, and the loss landscape of \eqref{eq:loss_func} exhibits the desirable properties that enable trainability found in classical networks~\cite{pmlr-v38-choromanska15,chaudhari2017energy}.

Unfortunately, unlike the classical setting, the performance of VQAs is often dominated by poor performance in the training procedure (see the Supplemental Information for a discussion). For one, VQAs tend to exhibit \emph{barren plataeus} when they are deep; namely, gradients of deep variational quantum circuits vanish exponentially with the problem size in many settings~\cite{mcclean2018barren,cerezo2020costfunctiondependent,napp2022quantifying}. Problematic training in this regime has also been studied beyond gradient descent~\cite{cerezo2020impact,arrasmith2021effect}.

Until recently, less was known about the trainability of VQAs in the \emph{shallow} model regime. Numerically, \cite{kiani2020learning,wiersema2020exploring} showed that randomly chosen variational landscapes typically have poor local minima, a result which was later proven in \cite{anschuetz2022critical} for nonlocal models using tools from random matrix theory. In a similar line of research, \cite{you2021exponentially} showed that for certain quantum variational ansatzes or quantum neural networks, there exist data sets and loss functions which induce exponentially many local minima in the loss landscape. \cite{liu2022,you2022} both showed that, in an \emph{overparameterized} regime, these models experience good local minima, though this transition to trainability typically occurs at an intractable number of parameters. \change{Finally, assuming the presence of a constant rate of noise per ansatz gate, \cite{deshpande2021} showed convergence of the loss landscape to the uniform distribution at a rate exponential in the circuit depth.} Many of these previous results on the untrainability of VQAs are summarized in Table~\ref{tab:previous_results}, along with a summary of our results which focus on the shallow, local regime.

\begin{table*}
    \begin{center}
        \begin{tabular}{c|c|c|c|c|c|c}
            Result & Dimension & Locality & Depth & \change{Worst} case? & Barren plateaus? & Poor minima?\\\hline
            \cite{mcclean2018barren} & $d$ & $2$ & $\operatorname{\Omega}\left(n^\frac{1}{d}\right)$ & \change{\xmark} & \cmark & ?\\\hline
            \cite{cerezo2020costfunctiondependent} & $1$ & $2$ & $\operatorname{\upomega}\left(\log\left(n\right)\right)$ & \change{\xmark} & \cmark & ?\\\hline
            \cite{napp2022quantifying} & $d$ & $2$ & $\operatorname{\upomega}\left(\log\left(n\right)^\frac{1}{d}\right)$ & \change{\xmark} & \cmark & ?\\\hline
            \cite{anschuetz2022critical} & N/A & $n$ & $\operatorname{\Omega}\left(1\right)$ & \change{\xmark} & \cmark/\xmark & \cmark/\xmark\\\hline
            \cite{you2021exponentially} & $d$ & $2$ & $\operatorname{\Omega}\left(1\right)$ & \change{\cmark} & ? & \cmark\\\hline\hline
            Our results & $d$ & $2$ & $\operatorname{\Omega}\left(1\right)$ & \change{\xmark} & \xmark & \cmark
        \end{tabular}
        \caption{A summary of previous results on the untrainability of variational quantum algorithms. A label of ``{\cmark}/{\xmark}'' denotes that the paper studied certain regimes where the phenomenon was present, and certain regimes where it was not. A label of ``?'' denotes that the phenomenon was not studied. \change{``Dimension'' indicates the locality structure of the ansatzes study. For instance, $\text{Dimension}=1$ denotes ansatzes with nearest-neighbor interactions for qubits on a line. ``Worst case'' denotes analysis performed with adversarial data.}\label{tab:previous_results}}
    \end{center}
\end{table*}

\subsection{Learning in the Statistical Query Framework}\label{sec:cqsc_learning}

\change{Quantum machine learning} algorithms are inherently noisy due to both unavoidable sources of error---such as \change{shot noise from} sampling outputs---or potentially correctable sources of error such as gate errors and state preparation noise. In such noisy settings, the \emph{statistical query} (SQ) model provides a useful framework for quantifying the complexity of learning a class of functions by considering how many query calls to a noisy oracle are needed to learn any function in that class (see the Supplemental Information for a brief review and history of SQ models)~\cite{goel2020superpolynomial,kearns1998efficient,szorenyi2009characterizing}. \change{In this setting, we consider the optimization of a risk of the form of \eqref{eq:empirical_risk}. We assume there is a target observable $\mM$ that we would like to learn on some distribution over states $\mathcal{D}$. We define a \emph{correlational statistical query} $\operatorname{qCSQ}(\mO,\tau)$, which takes in a bounded observable $\mO$ with $\| \mO \| \leq 1$ and a tolerance $\tau$ and returns a value in the range:
\begin{equation}
    \mathbb{E}_{\rho \sim \mathcal{D}}\left[\operatorname{Tr}(\mO\rho) \operatorname{Tr}(\mM\rho) -\tau \right] \leq \operatorname{qCSQ}(\mO,\tau) \leq \mathbb{E}_{\rho \sim \mathcal{D}}\left[\operatorname{Tr}(\mO\rho) \operatorname{Tr}(\mM\rho) +\tau \right].
\end{equation}
}
\change{Note that there are no guarantees on the form of the additive error other than it is within the tolerance $\tau$, and may for instance depend on the observable being queried $\bm{O}$.} Though SQ oracle calls may at first appear unrelated to variational algorithms, we show \change{in the Methods} that \change{many common variational optimizers in the presence of noise of the magnitude $\tau$ reduce to calls to an SQ oracle; for instance, commonly used first and second order optimization algorithms fall within the framework of the SQ model we consider. In the Methods, we also describe an analogous SQ model for learning unitaries.}

\begin{table*}[ht]
\small
\centering
\begin{tabular}{lc}
Setting ($n$ qubits, $L$ layers)  & Query complexity ($\beta < 1/2^*$) 
\\ \hline
\multirow{1}{*}{\parbox{7cm}{ $L=1$, global measurement, single qubit gates}} 
& $2^{\operatorname{\Omega}(n)}$ if $\tau \geq 3^{-\beta n}$
\\[0.1cm]
\multirow{2}{*}{\parbox{7cm}{ $L=\ceil{\log_2 n}$, single qubit measurement, global \change{$1$- and} $2$-local gates }} 
& $2^{\operatorname{\Omega}(n)}$ if $\tau \geq 4^{-\beta n}$
\\ \\[0.1cm]
\multirow{2}{*}{\parbox{7cm}{$L \ll n$, single qubit measurement, neighboring \change{$1$- and} $2$-local gates on \change{a} $d$-dim. lattice}}  
& $2^{\operatorname{\Omega}(L^d)}$ if $\tau = \operatorname{\Omega}(1)^{**}$
\\ \\[0.1cm]
\multirow{1}{*}{\parbox{7cm}{ $L=1$, single qubit gates\change{, unitary learning}}} 
& $2^{\operatorname{\Omega}(n)}$ if $\tau \geq 4^{-\beta n}$ \\ \hline
\multicolumn{2}{l}{$^*$ Technically, we require $\beta = 1/2 - \operatorname{\Omega}(1)$; $^{**}$ $\tau = 2^{-\operatorname{\upomega}\left(\min\left(2L,n^{1/d}\right)^d\right)}$ is sufficient.}
\end{tabular}
\caption{Relatively simple classes of functions require exponentially many statistical queries to learn using any naive algorithm that reduces to statistical queries. The table above quantifies the number of queries needed to identify a target function in the function class, over a distribution of states that forms a $2$-design and with queries that have tolerance $C_{max} \tau$ (query tolerance lower bounded by a constant times $C_{max}$ suffices in all cases). }
\label{tab:sq_dim}
\end{table*}

To quantify the hardness of learning variational circuits, we consider the task of learning certain function classes generated by shallow variational circuits over a distribution of inputs $\mathcal{D}$ which forms a $2$-design. Our results also generally hold for distributions that are uniform over states in the computational basis, recovering the statistical query setting for classical Boolean functions. \Cref{tab:sq_dim} summarizes the number of queries needed to learn various function classes which are generated by variational circuits, with proofs deferred to the Supplemental Information. In \change{all settings we consider}, an exponential number of queries (in either $n$ or the light cone size) are needed to learn simple classes, such as the class of functions generated by single qubit gates followed by a fixed global measurement. This hardness intuitively arises because each individual query can only obtain information about a few of the exponentially many orthogonal elements in the function class. \change{More formally, we lower bound the SQ dimension (defined in the Methods) of the function classes considered in Table~\ref{tab:sq_dim} to show our query lower bounds.}

\change{Our hardness results hold for \emph{any} target observable $\bm{M}$, as long as the learning setting is one we consider in Table~\ref{tab:sq_dim}. Furthermore, our results hold for \emph{any} constant error $\tau$ in the statistical queries; indeed, the majority of our results hold even if this noise were only exponentially small in the problem size. For instance, training via gradient descent where the gradient is estimated using polynomially many samples fits into this framework immediately just from the induced shot noise.}

In a more positive light, learning local Hamiltonians generated by shallow depth circuits can potentially be efficiently performed as the complexity grows exponentially only with locality or depth in this setting. In fact, prior results have provably shown that certain classes of Hamiltonians are efficiently learnable using properly chosen algorithms \cite{anshu2021sample,bairey2019learning}. Nevertheless, this does not correspond to efficient learnability of the ground state of a given Hamiltonian\change{,} as learnability of the properties of a Hamiltonian is not the same as the learnability of its ground state. \change{Indeed, we will see in Sec.~\ref{sec:loss_landscapes} that typically, even in this setting, learning the ground state of such a local Hamiltonian is difficult.}

Our hardness results do not indicate that simple classes of functions like those generated by single qubit rotations are hard to learn for \emph{any} algorithm, but only those whose steps reduce to statistical queries. For example, the class of Pauli channels is not learnable in the SQ setting, but there exist simple, carefully constructed, algorithms which can learn Pauli channels \cite{chen2022quantum,huang2021information,gollakota2022hardness}. This is analogous to the classical setting where parity functions are hard to learn in the noisy SQ setting, but efficient to learn using simple linear regression \cite{kearns1998efficient}. Similarly, the related work of \cite{hinsche2021learnability} showed that output distributions of Clifford circuits can be hard to learn using statistical queries, but efficient using a technique that resorts to linear regression on a matrix formed from samples of the overall distribution. More loosely, our results provide support to the basic maxim that algorithms which apply too broadly will work very rarely \cite{wolpert1997no}\change{; more careful construction of learning algorithms tailored to the problem at hand is generally necessary}. One straightforward way to avoid the hardness of the SQ setting is to construct algorithms whose basic steps do not reduce to statistical queries, e.g. via the construction of non-global metrics \cite{kiani2021,huang_power_2021,khatri2019quantum}. However, such a fix is by no means guaranteed to avoid the more general issues of poor landscapes and noise that also make learning in the SQ setting so difficult, as we \change{now examine}.

\subsection{Loss Landscapes of Local Variational Quantum Algorithms}\label{sec:loss_landscapes}
\change{We now consider the trainability of VQAs in the noise free regime, beyond optimization algorithms that reduce to statistical queries. Though we are unable to prove the very strong no-go results proved in the SQ framework, we are able to show that the loss landscapes of \emph{typical} local variational algorithms with Hamiltonian agnostic ansatzes are unamenable to optimization. We achieve this by showing that typically, the loss landscapes of shallow, local VQAs are swamped with poor local minima.}

As \change{discussed in Table~\ref{tab:previous_results}}, it is already known that deep Hamiltonian agnostic ansatzes are typically untrainable due to the presence of barren plateaus~\cite{mcclean2018barren,cerezo2020costfunctiondependent,napp2022quantifying}; hence, here we focus on shallow ansatzes. Previous results~\cite{anschuetz2022critical} have also shown that shallow, \emph{nonlocal} models \change{are untrainable, by showing that the scrambling of variational ansatzes over the entire system in these instances induce poor local minima. These techniques were not extendable to shallow, \emph{local} ansatzes, however, which do not scramble globally.}

\change{Instead, here, we show that ansatzes that \emph{approximately} scramble \emph{locally} are difficult to train. As we will later show, this includes common classes of variational ansatzes, such as Hamiltonian agnostic checkerboard ansatzes on a $d$-dimensional lattice. We show that this approximate, local scrambling suffices to imply that the loss landscapes of these VQAs are close to those of \emph{Wishart hypertoroidal random fields} (WHRFs). These are random fields parameterized by $l,m$ of the form:
\begin{equation}
    F_{\text{WHRF}}\left(\bm{w}\right)=m^{-1}\sum\limits_{i,j=1}^{2^l}w_i J_{i,j}w_j,
\end{equation}
where $\bm{J}$ is drawn from a Wishart distribution with $m$ degrees of freedom, and $\bm{w}$ are points on a certain embedding of the hypertorus $\left(S^1\right)^{\times l}$ in $\mathbb{R}^{2^l}$. We demonstrate this convergence via new techniques, directly bounding the error in the joint characteristic function of the function value, gradient, and Hessian components of the variational loss from those of WHRFs. As the typical loss landscapes of WHRFs are known given these random variables (see Methods for a summary), by demonstrating sufficient convergence of these random variables to those of WHRFs, we will be able to infer the distribution of critical points for local VQAs.}

To begin, we take our (assumed traceless) problem Hamiltonian to have Pauli decomposition:
\begin{equation}
    \bm{H}=\sum\limits_{i=1}^A\alpha_i\bm{P_i},
\end{equation}
and for simplicity scale and shift the loss landscape of \eqref{eq:loss_func} to be of the form:
\begin{equation}
    \change{\widehat{\mathcal{R}}_{\text{VQE}}}\left(\bm{\theta}\right)=1+\left\lVert\bm{\alpha}\right\rVert_1^{-1}\sum\limits_{i=1}^A\alpha_i\bra{\bm{\theta}}\bm{P_i}\ket{\bm{\theta}},
    \label{eq:practical_loss_func_inf}
\end{equation}
where $\bm{\alpha}$ is the vector of all $\alpha_i$ and the ansatz $\ket{\bm{\theta}}$ is as given in \eqref{eq:ansatz}. As this ansatz is assumed to be shallow and local, we assume that the reverse light cone of each $\bm{P_i}$ under the ansatz is of size $l\ll n$.

As in most analytic treatments of Hamiltonian agnostic VQAs, we consider certain randomized classes of ansatzes~\cite{mcclean2018barren,cerezo2020costfunctiondependent,napp2022quantifying,anschuetz2022critical}. Roughly, we assume that in a local region around each measured Pauli observable $\bm{P_i}$, the ansatz is an $\epsilon$-approximate $t$-design; that is, its first $t$ moments are $\epsilon$-close to those of the Haar distribution. This is a much weaker assumption than global scrambling \change{of the ansatz}. For instance, for $\bm{P_i}$ of constant weight, such \change{approximately locally scrambling circuits include} constant depth local circuits with random local gates~\cite{harrow2018approximate}. We discuss in more detail when this assumption holds in practice \change{when specializing to common variational quantum learning scenarios}, and defer technical details to the Supplemental Information.

Our main result, informally, is that the random field given by \eqref{eq:practical_loss_func_inf} under this approximate, local scrambling assumption converges in distribution to that of a WHRF. The formal statement and derivation of this result are given in the Supplemental Information, where we also lay out our assumptions more explicitly.
\begin{theorem}[Approximately locally scrambled variational loss functions converge to WHRFs, informal]
    Let
    \begin{equation}
        m\equiv\frac{\left\lVert\bm{\alpha}\right\rVert_1^2}{\left\lVert\bm{\alpha}\right\rVert_2^2}2^{l-1}
    \end{equation}
    be the \emph{degrees of freedom} parameter. Assume $q\log\left(q\right)=\operatorname{o}\left(m\right)$, where $q$ is the number of ansatz parameters in the reverse light cone of each $\bm{P_i}$. Then, the distribution of \eqref{eq:practical_loss_func_inf} \change{and its first two derivatives are} equal to \change{those} of a WHRF
    \begin{equation}
        F_{\text{WHRF}}\left(\bm{\theta}\right)=m^{-1}\sum\limits_{i,j=1}^{2^l}w_i J_{i,j}w_j
        \label{eq:whrf}
    \end{equation}
    with $m$ degrees of freedom, up to an error in distribution on the order of $\operatorname{\tilde{O}}\left(\operatorname{poly}\left(\frac{1}{t}+\epsilon+\exp\left(-l\right)\right)\right)$. Here, $\bm{w}$ are points on the hypertorus $\left(S^1\right)^{\times l}$ parameterized by $\bm{\tilde{\theta}}$, where $\tilde{\theta}_i$ is the sum of all $\theta_j$ on qubit $i$.
    \label{thm:app_scramb_vqe_convs_whrfs_inf}
\end{theorem}

We interpret this result as the degrees of freedom $m$ of the model being given by roughly the sum of the \emph{local} Hilbert space dimensions of the reverse light cones of terms in the Pauli decomposition of $\bm{H}$. We interpret this as the \emph{local underparameterization} of the model, to be contrasted with the \emph{global underparameterization} interpretation when $m$ is exponentially large in $n$. \change{Using known properties of the loss landscapes of WHRFs (see Methods), we are then able to prove the following result on the loss landscapes of local VQAs:
\begin{corollary}[Shallow, local VQAs have poor loss landscapes, informal]
    Let $\widehat{\mathcal{R}}_{\text{VQE}}$ be a local VQA loss function of the form of \eqref{eq:practical_loss_func_inf}. Assume all coefficients $\alpha_i$ of the Pauli decomposition of $\bm{H}$ are $\operatorname{\Theta}\left(1\right)$, and
    \begin{equation}
        l\log\left(n\right)+q\log\left(q\right)=\operatorname{o}\left(2^l A\right).\label{eq:cor_cond}
    \end{equation}
    Then $\widehat{\mathcal{R}}_{\text{VQE}}$ has a fraction superpolynomially small in $n$ of local minima within any constant additive error of the ground state energy.\label{cor:poor_minima_vqas}
\end{corollary}
}

\change{Optimizing loss landscapes where only a superpolynomially small (in $n$) fraction of the local minima are near the global minimum in energy is expected to be difficult. Indeed, algorithms such as gradient descent would then expect to have to be restarted a superpolynomial number of times before finding a good approximation to the global minimum; we also give heuristic reasons why this should continue to be true for other local optimizers in the Supplemental Information. Our results stand in stark contrast with the loss landscapes typically found in classical machine learning, where almost all local minima closely approximate the global minimum in function value~\cite{pmlr-v38-choromanska15,chaudhari2017energy}.}

\change{In the shallow ansatz regime---where $q,l=\operatorname{O}\left(\operatorname{polylog}\left(n\right)\right)$---and assuming an extensive problem Hamiltonian such that $A=\operatorname{\Omega}\left(n\right)$, the condition given by \eqref{eq:cor_cond} is always satisfied. Interestingly, this is a regime where barren plateaus are known not to occur~\cite{napp2022quantifying}, demonstrating that poor local minima can give rise to poor optimization performance even when the loss function features large gradients. We now specialize to common variational quantum learning scenarios, and consider the implications of Corollary~\ref{cor:poor_minima_vqas}.}

\paragraph{Checkerboard ansatzes} First, let us consider $d$-dimensional checkerboard ansatzes of constant depth. Fix $p,t$ to be sufficiently large constants. We assume that the initial state forms an $\operatorname{O}\left(\frac{1}{\operatorname{poly}\left(t\right)}\right)$-approximate $t$ design on $l$ qubits around each Pauli observable of weight $k$; this can be done via a depth $p$, $d$-dimensional circuit of 2-local Haar random unitaries when $l=\operatorname{O}\left(\frac{\left(p+k\right)^d}{\operatorname{poly}\left(t\right)}\right)\geq k$ for some fixed polynomial in $t$~\cite{harrow2018approximate,haferkamp2022}. After this state preparation circuit, a traditional depth $\operatorname{\Theta}\left(l^{\frac{1}{d}}\right)$ (i.e. independent of $n$), $d$-dimensional, $n$ qubit checkerboard circuit is applied, with observable reverse light cones of size at greatest $l$. By \change{Corollary~\ref{cor:poor_minima_vqas}}, these variational ansatzes are untrainable due to poor local minima, yet by the results of \cite{napp2022quantifying} do not suffer from barren plateaus.

One interesting consideration is extending this result to ``traditional'' checkerboard ansatzes, without the special state preparation procedure we have considered. There, the $l=\operatorname{O}\left(\frac{\left(p+k\right)^d}{\operatorname{poly}\left(t\right)}\right)$ qubit local state is mixed, and our results therefore do not directly apply. However, we expect no reason for the mixedness of the initial state to improve training performance in any way. We validate this intuition numerically in Section~\ref{sec:numerics}.

\paragraph{Quantum convolutional neural networks} We  also consider a class of models similar to quantum convolutional neural networks (QCNNs)~\cite{cong2019} previously shown not to suffer from barren plateaus~\cite{pesah2021absence}. Though these models are in full generality trained on arbitrary loss functions, for learning various physical models the loss may take the form of \eqref{eq:loss_func}. QCNNs are defined by their measurement of a subset of qubits at periodic intervals, via so-called \emph{pooling layers}; for sufficiently deep (i.e. large constant depth) convolutional layers, then, at some point in the model, the number of remaining qubits will be sufficiently small such that the remaining convolutional layers are approximately scrambling. If one then assumes that the initial states are adversarially chosen such that they remain pure by this layer, this scenario reduces to the shallow checkerboard ansatz scenario, and once again we expect poor local minima by \change{Corollary~\ref{cor:poor_minima_vqas}}. Even if the initial states are not adversarially chosen and the input to the scrambling convolutional layers is mixed, we expect by similar intuition the model to remain untrainable; we see this numerically, for instance, in Section~\ref{sec:teacher_student}. We also see in Section~\ref{sec:teacher_student} that this poor training occurs even when training on loss functions beyond \eqref{eq:loss_func}.

\subsection{Numerical Results}\label{sec:numerics}

To numerically validate our theoretical findings, we perform numerical simulations showing that learning in various settings cannot be guaranteed unless exponentially many parameters are included in an ansatz. We only consider problems and ansatzes where the existence of a zero loss global minima is guaranteed to study whether or not optimizers can actually find the global minimum or a similarly good critical point. We parameterize all trainable 2-qubit gates in the Lie algebra of the 4-dimensional unitary group, and implement the resulting unitary matrix via the exponential map which is surjective and capable of expressing any local $4 \times 4$ unitary gate. In all cases, we perform simulations using calculations with computer precision and analytic forms of the gradient (see the Supplemental Information for more details). In practice, actual quantum implementations will be hampered by various sources of inefficiency such as the lack of an analogous method of backpropogation for calculating gradients, sampling noise, or even gate errors. Thus, our numerical analysis can be interpreted as a ``best case'' setting for quantum computation where we disregard such inefficiencies and focus solely on learnability. \change{In the Supplemental Information, we further study variations of the teacher-student learning and random variational quantum eigensolver (VQE) settings discussed here. We also consider the training performance of VQE in finding the ground state of a Heisenberg XYZ Hamiltonian~\cite{heisenberg1928}. Our supplemental results reinforce our findings here.}

\subsubsection{Teacher-Student Learning}
\label{sec:teacher_student}
One may conjecture that it is plausible to learn the class of functions generated by relatively shallow depth variational ``teacher'' circuits by parameterizing a shallow-depth ``student" circuit of the same form and training its parameters. In this so-called teacher-student setup, we are guaranteed the existence of a ``perfect'' global minimum since recovering the parameters of the teacher circuit achieves zero loss. Nevertheless, we showed earlier that such circuits are \change{typically} have many poor local minima, and are \change{always} hard to learn in the statistical query setting. Here, we provide numerical evidence of these findings \change{for the QCNN ansatz. Additional confirmation of these findings with a checkerboard ansatz is included in the Supplemental Information.}

\begin{figure*}[t]
    \centering
    \includegraphics[]{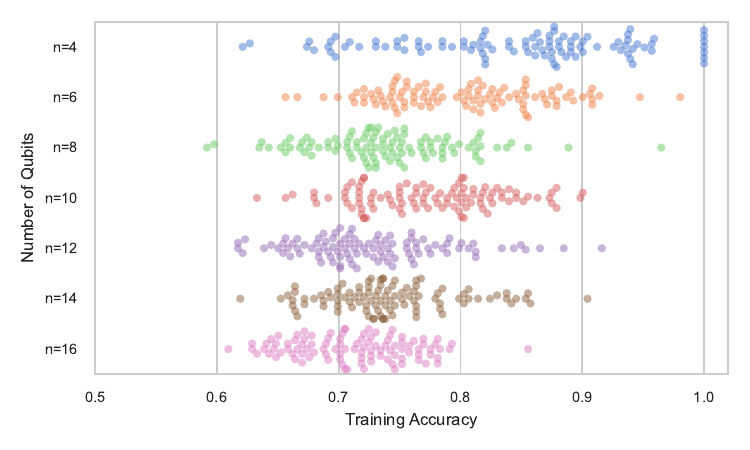}
    \caption{Teacher-student performance evaluation for the $n$ qubit QCNN. The student circuit is unable to learn the teacher circuit as the number of qubits grows, converging to a local minimum of the loss landscape. The existence of a global optimum is guaranteed as the teacher circuit is drawn from a random initialization of the same QCNN structure of the student circuit. Here, for a ranging number of qubits, 100 student circuits are trained to learn randomized teacher circuits of the same form and the resulting swarm plots of the final training accuracy are shown.}
    \label{fig:QCNN}
\end{figure*}

The quantum convolutional neural network (QCNN) presents an interesting test bed for our analysis since it has been shown in prior work to avoid barren plateaus~\cite{pesah2021absence}. Nevertheless, the QCNN, like other models, is riddled with poor local minima in generic learning tasks. For our analysis, we attempt to learn randomly generated quantum convolutional neural networks (QCNNs) with a parameterized QCNN of the same form. In the QCNN, both student and teacher circuits have parameterized 2-qubit gates at each layer followed by 2-qubit pooling layers (see Supplemental Information for more details). Each 2-qubit gate is fully parameterized in the Lie algebra of the unitary group. Networks are trained to predict the probability of the measurement of the last qubit in the teacher circuit. \change{In other words, the student network is trained on a classification problem defined by teacher network where, by construction, perfect classification accuracy is known to be achievable.} We benchmark performance with the classification accuracy, where a prediction is considered correct when it predicts the most likely measurement of the last qubit correctly. Networks are trained via the Adam optimizer \cite{kingma2014adam} to learn outputs of 512 randomly chosen computational basis states. QCNNs with $4$, $8$, $12$, and $16$ qubits have $32$, $48$, $64$, and $64$ trainable parameters, respectively.


\Cref{fig:QCNN} plots the final training accuracy achieved over $100$ random simulations for varying ranges of circuit sizes. For circuits with $4$ qubits, the training is sometimes successful, often achieving an accuracy above $85$ percent on the training dataset. However, as the number of qubits grows, even past $8$ qubits, the optimizer is unable to recover parameters which match the outputs of the teacher circuit. The results here show that the QCNN circuit---which has $\operatorname{O}(\log n)$ depth---still scrambles outputs to hinder learnability.

\subsubsection{Variational Quantum Eigensolvers}
\label{sec:numerics_VQE}
To analyze the performance of variational optimizers, we consider problems and ansatzes which are capable of recovering the global minimum. We aim to find the ground states of local Hamiltonians $\mH_t$ over $n$ qubits that take the form of single qubit Pauli $\mZ$ Hamiltonians conjugated by \change{$L^\ast$} layers of two alternating unitary operators $\mU_1$ and $\mU_2$ which are product unitaries on neighboring 2-local qubits:
\begin{equation}
    \mH_t = \left( \mU_2^\dagger \mU_1^\dagger \right)^{L^\ast} \left[ \sum_{i=1}^n \mZ_i \right] \left( \mU_1 \mU_2 \right)^{L^\ast} + n\mI.
    \label{eq:random_hamiltonian_form}
\end{equation}
The added identity matrix normalizes the Hamiltonian to have ground state with energy $0$. Since the ground state of $\sum_{i=1}^n \mZ_i$ is the state $\ket{1}^{\otimes n}$, we are guaranteed the existence of a global minima when using a checkerboard ansatz of at least depth $L^\ast$, since this ansatz can ``undo'' the conjugation by unitary operators. \change{In the remainder of this Section, we consider \eqref{eq:random_hamiltonian_form} with $L^\ast=4$.}

\begin{figure*}[!ht]
 \captionsetup[subfigure]{aboveskip=-1pt,belowskip=-3pt}  
 \begin{subfigure}{0.5\textwidth}
    \includegraphics[width=\linewidth]{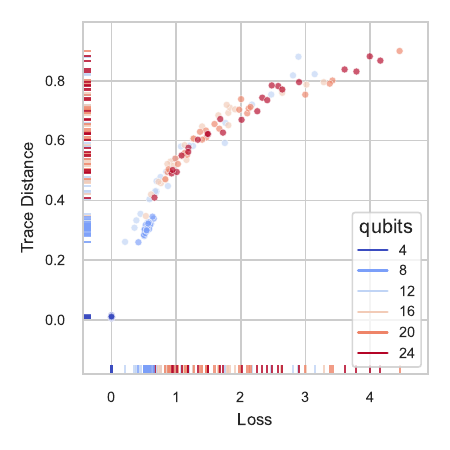}
    \caption{Optimization convergence} \label{fig:VQE_final_state_scatter}
  \end{subfigure}%
  \hspace*{\fill}   
 \begin{subfigure}{0.5\textwidth}
    \includegraphics[width=\linewidth]{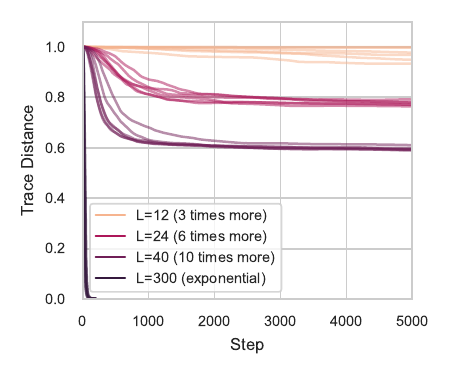}
    \caption{VQE overparameterization} \label{fig:VQE_overparameterization}
  \end{subfigure}%
  \hspace*{\fill}   
\caption{(a) Scatter plot of the final loss and trace distance of the VQE state after 30000 steps of gradient descent optimization shows that the algorithm converges to poorer local minima as the number of qubits grows. 24 simulations are performed for each value of $n$. The algorithm always succeeds at obtaining the ground state with $4$ qubits, but progressively struggles more with added qubits. (b) \change{The number of layers needed to guarantee convergence to the ground state empirically grows exponentially with the number of qubits. Here, we consider 4-layer Hamiltonians of the form of \eqref{eq:random_hamiltonian_form} on $14$ qubits where the number of layers $L$ in the ansatz is varied. When the ansatz has $300$ layers---enough that the number of ansatz parameters is larger than the explored Hilbert space dimension---the model successfully converges to the ground state, rather than remaining stuck in a poor local minimum. }}\label{fig:VQE_final_state}
\end{figure*}

We measure the performance of optimization with two metrics. The first is the loss function itself, which is the average energy $\bra{\psi}\mH_t\ket{\psi}$ of the VQE ansatz state $\ket{\psi}$ for the given Hamiltonian $\mH_t$. The second is the trace distance to the ground state $\ket{\phi_g}$ of $\mH_t$, equal to $\|\ket{\phi_g}\bra{\phi_g}-\ket{\psi}\bra{\psi}\|_1/2$. Both of these metrics converge to zero at the global minimum.

We first aim to learn the ground state using a checkerboard ansatz by performing vanilla gradient descent on $L=L^\ast=4$ parameterized layers, equal in depth to the Hamiltonian conjugation circuit and thus capable of recovering the ground state. In \Cref{fig:VQE_final_state_scatter}, we plot the final values of the loss and trace distance for 24 randomly initialized VQE problems for a number of qubits ranging from $4$ to $24$. Similar results are observed when using more advanced optimizers such as Adam (see the Supplemental Information) \cite{kingma2014adam}. Consistent with our theoretical findings, convergence clusters around local minima far from the ground state, particularly as the number of qubits grows. 

\change{Our theoretical results also imply the difficulty of training beyond a finite fraction of the ground state energy in a VQE setting. \Cref{fig:VQE_overparameterization} illustrates this phenomenon when performing optimization on a $14$ qubit ansatz. As more parameters are added to the ansatz via increasing its depth $L$, the VQE algorithm performs better, but it is not until the number of parameters is exponential in the problem size that convergence to a global minimum (or even within a small additive error of the global minimum) is guaranteed. This is true even though the ansatz is capable of expressing the ground state at $L=4$. Simulations here are performed as before on random $L^\ast=4$ Hamiltonians of the form of \eqref{eq:random_hamiltonian_form}.}

\section{Discussion}\label{sec:conclusion}

Though variational quantum algorithms---and quantum machine learning models in general---have been cited as perhaps the most promising use case for quantum devices in the near future~\cite{Preskill2018quantumcomputingin}, theoretical guarantees of their \change{training} performance have been sparse. Here, we have excluded a wide class of variational algorithms by showing that \change{in many settings}, they are in fact \emph{not} trainable. We showed this in two different frameworks: first, in Section~\ref{sec:cqsc_learning}, we studied various classes of quantum models in the statistical query framework. We showed that \change{in the presence of} noise, exponentially many queries in the problem size are needed for these models to learn. As a complementary approach, we also examined the \emph{typical} loss landscapes of variational quantum algorithms in the \emph{noiseless} setting in Section~\ref{sec:loss_landscapes}, and showed that even at constant depth these models can have a number of poor local minima superpolynomially large in the problem size. We also numerically confirmed these results for a variety of problems in Section~\ref{sec:numerics}. These results go beyond the typical studies on the presence of barren plateaus, as \change{many of the models we study here} have gradients vanishing only polynomially quickly in the problem size. Our work demonstrates that showing that barren plateaus are not present in a model does not necessarily vindicate it as trainable.

\change{These results, though they exclude a wide variety of variational quantum algorithms, still leave room for hope in the usefulness of these algorithms. Particularly, our analysis in the noiseless setting of landscapes of variational quantum algorithms focuses on very general, Hamiltonian agnostic ansatzes; in various instances, more focused ansatzes may be trainable. For instance, as previously shown in \cite{farhi2019}, for certain classes of problems the quantum approximate optimization algorithm (QAOA)~\cite{farhi2014} is provably able to outperform the best unconditionally proven classical algorithms, even when taking into account the training of the model. This is due to \emph{parameter concentration}, where the global optimum for small problem instances is close to the global optimum for large problem instances~\cite{brandao2018}. These results demonstrate the power of \emph{good model initialization} in variational quantum algorithms: even if the total variational landscape is swamped with poor local minima, good initialization may ensure that the optimizer begins in the region of attraction of the global minimum. Though this is perhaps most relevant for the variational quantum eigensolver (VQE)~\cite{peruzzo2014variational} and QAOA~\cite{farhi2014}, where there exists physical intuition for potentially performant parameter initializations, in more traditional machine learning settings this may manifest as good performance on certain inputs to the model.}

\change{Variationally studying models with many symmetries may also avoid our poor performance guarantees. Intuitively, our results here are the consequence of \emph{underparameterization}. Namely, unless the ansatz is parameterized such that the number of parameters grows with the (local) Hilbert space dimension, the model is not trainable. Typically, this Hilbert space dimension is exponentially larger than the number of parameters the ansatz uses to explore it. However, if the model is heavily constrained by symmetries, this dimension might be much smaller. Such models were studied numerically in \cite{larocca2021theory,wiersema2020exploring}, where it was shown that certain variational quantum algorithms optimize efficiently. Though often these models can be solved classically when the symmetries are known, these symmetries may not be known \textit{a priori}. Indeed, one may be able to test for the presence of symmetries in a given model by studying whether associated variational quantum algorithms are trainable. Similar to these general symmetry considerations, known structure in the problem may also allow one to build up hierarchical ansatzes that are able to be trained sequentially. We leave further investigation in these directions to future work.}

\change{Finally, though many variational models fit the framework of \eqref{eq:loss_func}, there exist other settings of variational quantum algorithms. One class of such models includes \emph{quantum Boltzmann machines}, which attempt to model given quantum states via the training of quantum Gibbs states~\cite{PhysRevX.8.021050}. When the full quantum Gibbs state is observed, it is known that these models are efficiently trainable~\cite{anshu2021sample}, and numerically it is known that these models are trainable even when the full state is not observed~\cite{PhysRevX.8.021050,PhysRevA.96.062327}. Furthermore, though in full generality preparing quantum Gibbs states is difficult, state preparation has been shown to be efficient in certain regimes relevant to machine learning~\cite{PhysRevA.96.062327,anschuetz2019,zoufal2021variational}, potentially giving an end-to-end trainable quantum machine learning model. We leave further analytical investigation on the training landscapes of quantum Boltzmann machines to future work.}

Our results contribute to the already vast library of literature on the trainability of variational quantum models in further culling the landscape of potentially trainable quantum models. We hope these results have the effect of focusing research efforts toward classes of models that have the potential for trainability, and whittle down the search for practical use cases of variational quantum algorithms.

\section{Methods}

\subsection{The Statistical Query Learning Framework}

We give a brief overview of the classical SQ model here, and provide a more detailed review in the Supplemental Information. Given an input and output space $\mathcal{X}$ and $\mathcal{Y}$, let $\mathcal{D}$ be a joint distribution on $\mathcal{X} \times \mathcal{Y}$. In the classical SQ model, one queries the SQ model by inputting a function $f$ and receiving an estimate of $\mathbb{E}_{(x,y)\sim\mathcal{D}}[f(x,y)]$ within a given tolerance $\tau$. As an example, one can query a loss function $\ell$ for a model $m_{\bm \theta}$ with parameters $\bm \theta$ by querying the function $\ell(m_{\bm \theta}(x) ,y)$. A special class of statistical queries are inner product queries where query functions $g$ are defined only on $\mathcal{X}$ and the correlational statistical query returns an estimate of $\mathbb{E}_{(x,y)\sim\mathcal{D}}[g(x) \cdot y]$ within a specified tolerance $\tau$.

\change{In detail, the SQ models we consider take the forms below.}
\paragraph{Quantum correlational statistical query (qCSQ)} Assume there is a target observable $\mM$ that we would like to learn on some distribution over states $\mathcal{D}$. Applying the correlational SQ model to the quantum setting, we define the query $\operatorname{qCSQ}(\mO,\tau)$ which takes in a bounded observable $\mO$ with $\| \mO \| \leq 1$ and a tolerance $\tau$ and returns a value in the range:
\begin{equation}
    \mathbb{E}_{\rho \sim \mathcal{D}}\left[\operatorname{Tr}(\mO\rho) \operatorname{Tr}(\mM\rho) -\tau \right] \leq \operatorname{qCSQ}(\mO,\tau) \leq \mathbb{E}_{\rho \sim \mathcal{D}}\left[\operatorname{Tr}(\mO\rho) \operatorname{Tr}(\mM\rho) +\tau \right].
\end{equation}

\paragraph{Quantum unitary statistical query (qUSQ)} In the unitary compilation setting, one aims to learn a target unitary transformation $\mU_*$ over a distribution $\mathcal{D}$ of input/output pairs of that unitary transformation. Here, the oracle $\operatorname{qUSQ}(\mV, \tau)$ takes in a unitary matrix $\mV$ and a tolerance $\tau$ and returns a value in the range:
\begin{equation}
    \mathbb{E}_{\rho \sim \mathcal{D}}\left[\Re[\operatorname{Tr}(\mU_*^\dagger \mV \rho)] -\tau \right] \leq \operatorname{qUSQ}(\mV,\tau) \leq \mathbb{E}_{\rho \sim \mathcal{D}}\left[\Re[\operatorname{Tr}(\mU_*^\dagger \mV \rho)] +\tau \right].
\end{equation}
Importantly, if $\mathcal{D}$ is a $1$-design over $n$ qubit states, then the above can be simplified using the formula $\mathbb{E}_{\rho \sim \mathcal{D}}\left[\Re[\operatorname{Tr}(\mU_*^\dagger \mV \rho)] \right] = 2^{-n}  \Re\left[ \Tr(\mU_*^\dagger \mV)\right]$ (see proof in the Supplemental Information). Queries to $\operatorname{qUSQ}$ are related to performing a Hadamard test \cite{aharonov2009polynomial}, also a common subroutine in variational algorithms \cite{wang2021variational}.

The queries above take the forms of inner products, with $\langle \mM_1, \mM_2 \rangle_{\mathcal{D}} =\mathbb{E}_{\rho \sim \mathcal{D}}\left[\operatorname{Tr}(\mM_1\rho) \operatorname{Tr}(\mM_2\rho) \right]$ and $\langle \mU_1, \mU_2 \rangle_{\mathcal{D}} =\mathbb{E}_{\rho \sim \mathcal{D}}\left[\Re\left[\operatorname{Tr}(\mU_1^\dagger \mU_2 \rho)\right] \right]$. The inner products also induce corresponding $L_2$ norms: $\| \mM \|_{\mathcal{D}} = \sqrt{ \langle \mM, \mM \rangle_{\mathcal{D}}}$. As the magnitude of this norm can change with the dimension, we introduce the quantity $C_{max}$ to denote the maximum value a query can take for any target observable in the $\operatorname{qCSQ}$ model, i.e. $C_{max}=\max_{\mM:\|\mM\|\leq 1} \|\mM\|_{\mathcal{D}}^2$. For fair comparison, we quantify noise tolerances and hardness bounds with respect to $C_{max}$. Note that for the $\operatorname{qUSQ}$ model $C_{max} = 1$, but in the $\operatorname{qCSQ}$ model, $C_{max}$ can decay with the number of qubits under for example the Haar distribution of inputs. 

A statistical query algorithm \emph{learns} a function class if it can output a unitary or observable that is close to any target in that class.
\begin{definition}[$\operatorname{qCSQ}/\operatorname{qUSQ}$ learning of hypothesis class]
A given algorithm using only statistical queries to $\operatorname{qCSQ}$ ($\operatorname{qUSQ}$) successfully learns a hypothesis class $\mathcal{H}$ consisting of observables $\mM, \| \mM \| \leq 1$ (unitaries $\mU$) up to $\epsilon$ error if it is able to output an observable $\mO$ (unitary $\mV$) which is $\epsilon$-close to the unknown target observable $\mM \in \mathcal{H}$ ($\mU \in \mathcal{H}$) in the $L_2$ norm, i.e., $\| \mM - \mO \|_{\mathcal{D}} \leq \epsilon$ ($\| \mU - \mV \|_{\mathcal{D}} \leq \epsilon$). 
\end{definition}

The \emph{statistical query dimension} quantifies the complexity of a hypothesis class $\mathcal{H}$ and is related to the number of queries needed to learn functions drawn from a class, as summarized in Theorem~\ref{thm:sq_lower_bound}.
\begin{restatable}[Statistical query dimension \cite{blum1994weakly,reyzin2020statistical}]{definition}{SQdimDefinition}
For a distribution $\mathcal{D}$ and concept class $\mathcal{H}$ where $\|\mM\|_{\mathcal{D}}^2\leq C_{max}$ for all $\mM \in \mathcal{H}$, the statistical query dimension ($\operatorname{SQ-DIM}_{\mathcal{D}}(\mathcal{H})$) is the largest positive integer $d$ such that there exists $d$ observables $\mM_1, \mM_2, \dots, \mM_d \in \mathcal{H}$ such that for all $i \neq j:$ $| \langle \mM_i, \mM_j \rangle_{\mathcal{D}} | \leq C_{max}/d$.
\end{restatable}

\begin{restatable}[Query complexity of learning~\cite{szorenyi2009characterizing,blum1994weakly}]{theorem}{SQlowerbound}
\label{thm:sq_lower_bound}
Given a distribution $\mathcal{D}$ on inputs and a hypothesis class $\mathcal{H}$ where $\| \mM \|_{\mathcal{D}}^2 \leq C_{max}$ for all $\mM \in \mathcal{H}$, let $d=\operatorname{SQ-DIM}_{\mathcal{D}}(\mathcal{H})$ be the statistical query dimension of $\mathcal{H}$. Any $\operatorname{qCSQ}$ or $\operatorname{qUSQ}$ learner making queries with tolerance $C_{max}\tau$ must make at least $(d \tau^2 -1 )/2 $ queries to learn $\mathcal{H}$ up to error $C_{max}\tau$.
\end{restatable}
Since our setting differs slightly from the standard classical setting~\cite{szorenyi2009characterizing,blum1994weakly}, we include a proof of the above in the Supplemental Information. For example, if the hypothesis class is rich enough to be able to express any $n$-qubit Pauli observable, then the statistical query dimension of that class is at least $4^n$ over the Haar distribution of inputs since Pauli observables are all orthogonal. This forms the basis for our resulting proofs of hardness, summarized in \Cref{tab:sq_dim} and proved in the Supplemental Information.

\subsubsection{Noisy Optimization as Statistical Queries}\label{sec:noisy_opt_as_stat_quer}

\begin{figure}
    \centering
    \includegraphics[width=0.5\textwidth]{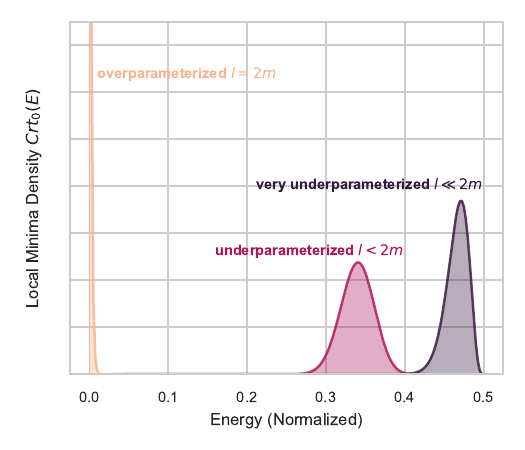}
    \caption{Plot of the asymptotic distribution of local minima of WHRFs with $m$ degrees of freedom on the $l$-torus in: the extremely underparameterized regime, where $l\ll 2m$; the moderately underparameterized regime, where $l$ is a finite fraction of $2m$; and at the critical overparameterization regime, where $l=2m$. Here, the energy is scaled and shifted as per \eqref{eq:practical_loss_func_inf} so that global minima have zero energy. In the underparameterized regime, only a fraction $\sim\exp\left(-m\right)$ of the critical points are within any constant additive error of the global minimum. In the overparameterized regime, local minima are exponentially concentrated at the global minimum. \label{fig:whrf_minima_dist} }
\end{figure}
Analogous to work in classical machine learning \cite{goel2020superpolynomial}, one can perform noisy gradient descent as a series of statistical queries. As an example, consider the task of learning a target Hamiltonian $\mM$ by constructing a variational Hamiltonian $\mH(\bm \theta) = \mU(\bm \theta)^\dagger \mH \mU(\bm \theta)$ with parameterized Pauli rotations and minimizing the mean squared error between expectations of $\mM$ versus $\mH(\bm \theta)$ over a distribution of states $\mathcal{D}$. Our loss function is
\begin{equation}
    \mathcal{L}(\bm \theta) = \mathbb{E}_{\rho \sim \mathcal{D}} \left[\left( \Tr[\mM \rho] - \Tr[\mH(\bm \theta) \rho] \right)^2  \right].
\end{equation}
\change{The parameter shift rule \cite{schuld2019evaluating} provides a means to calculate the partial derivative of a function $f(\mu)$ with respect to a parameter $\mu$ applied as a parameterized quantum gate $e^{-i\mu \mG}$ by calculating the function itself at two shifted coordinates. For example, for parameterized Pauli gates ($\mG \in \frac{1}{2}\{\mZ,\mX,\mY\}$), this takes the form:
\begin{equation}
    \frac{\partial}{\partial \mu} f(\mu) = \frac{1}{2}\left[f\left(\mu + \frac{\pi}{2}\right) - f\left(\mu - \frac{\pi}{2}\right) \right].
\end{equation}
}

By applying the parameter shift rule \cite{schuld2019evaluating}, we can evaluate the gradient of the loss with respect to parameter entry $\bm \theta_i$ as
\begin{equation}
\begin{split}
    \frac{\partial}{\partial \bm \theta_i} \mathcal{L}(\bm \theta) &=  \mathbb{E}_{\rho \sim \mathcal{D}}\left[ \Bigl( \Tr[\mH(\bm \theta) \rho] - \Tr[\mM \rho] \Bigr) \Bigl(  \Tr[\mH(\bm \theta^+) \rho] - \Tr[\mH(\bm \theta^-) \rho] \Bigr) \right],
\end{split}
\end{equation}
where $\bm \theta^+$ and $\bm \theta^-$ are the values of the parameters shifted at the $i$-th entry according to the parameter shift rule for the gradient. The quantity $\mathbb{E}_{\rho \sim \mathcal{D}}\left[  \Tr[\mH(\bm \theta) \rho] \bigl(  \Tr[\mH(\bm \theta^+) \rho] - \Tr[\mH(\bm \theta^-) \rho] \bigr) \right] $ can be directly evaluated without statistical queries, and the quantity $\mathbb{E}_{\rho \sim \mathcal{D}}\left[  \Tr[\mM \rho] \bigl(  \Tr[\mH(\bm \theta^+) \rho] - \Tr[\mH(\bm \theta^-) \rho] \bigr) \right]$ can be evaluated using $2$ statistical queries to $\operatorname{qCSQ}$ where the tolerance $\tau$ accounts for the noise in the estimate. 

\change{As a second example, this time in the unitary compiling setting of $\operatorname{qUSQ}$, we can evaluate the commonly used procedure of measuring the inner product or average fidelity of $n$-qubit states between a target unitary $\mU_*$ and a variationally chosen unitary $\mV(\bm \theta)$ using statistical queries analogous to a swap test on actual quantum hardware \cite{mitarai2018quantum,buhrman2001quantum,khatri2019quantum,beer2020training}. With slight abuse of notation, let $\ket{\phi} \sim \mathcal{D}$ denote a distribution over pure states which forms a $2$-design. Then via averaging over $2$-designs (see Supplemental Information for details), the average fidelity equals
\begin{equation}
    \mathbb{E}_{\ket{\phi} \sim \mathcal{D}} \left[ F\Bigl(\mU_* \ket{\phi}, \mV(\bm \theta) \ket{\phi} \Bigr) \right] = \mathbb{E}_{\ket{\phi} \sim \mathcal{D}} \left[ | \bra{\phi} \mV(\bm \theta)^\dagger \mU_*  \ket{\phi} |^2 \right] = \frac{2^{-n} \left| \Tr( \mV(\bm \theta)^\dagger \mU_* ) \right|^2 + 1 }{2^n + 1}.
\end{equation}
Note, that the key quantity $\left| \Tr( \mV(\bm \theta)^\dagger \mU_* ) \right|^2 = \Re[ \Tr( \mV(\bm \theta)^\dagger \mU_* ) ]^2 + \Re[ i \Tr( \mV(\bm \theta)^\dagger \mU_* ) ]^2 $ can be evaluated up to a desired tolerance using statistical queries $\operatorname{qUSQ}(\mV(\bm \theta), \tau)$ and $\operatorname{qUSQ}(i\mV(\bm \theta), \tau)$. }

One important caveat must be noted that in the SQ setting, learning must succeed for all values of the query within the given tolerance $\tau$. Noise in quantum settings, which can arise from sampling a finite data set, gate error, state preparation error, measurement sampling noise, or other means does not exactly coincide with the assumed tolerance of an SQ model. Nevertheless, though noise during optimization may appear unnatural in classical settings, such noise in quantum settings is rather endemic and the SQ model allows one to rigorously analyze the complexity of learning \change{in the presence of} noise.

\subsection{The Loss Landscapes of Wishart Hypertoroidal Random Fields}

The loss landscapes of \emph{Wishart hypertoroidal random fields} (WHRFs; see the Supplemental Information for a brief review) are known~\cite{anschuetz2022critical} to exhibit a computational phase transition governed by the order parameter
\begin{equation}
    \gamma=\frac{l}{2m},
\end{equation}
called the \emph{overparameterization ratio}. Here, $l$ is the number of parameters of the WHRF, and $m$ its degrees of freedom (see the Supplemental Information). When $\gamma\ll 1$ (the \emph{underparameterized regime}), WHRFs exhibit poor local minima and thus are essentially untrainable; when $\gamma\geq 1$ (the \emph{overparameterized regime}), however, essentially all local minima of a WHRF are close to the global minimum in function value. \change{More specifically, when $\gamma=\operatorname{o}\left(\frac{1}{\log\left(n\right)}\right)$, a superpolynomially small (in $n$) fraction of the local minima are within any constant additive energy error to the global minimum. When restoring units to the variational risk of \eqref{eq:practical_loss_func_inf}, this is an error extensive in the problem size.} The asymptotic expression of the distribution of local minima is also known, which is given by (up to a normalization factor):
\begin{equation}
    \operatorname{Crt}_0\left(E\right)\sim\ce^{-mE}E^{m-l/2}\left(1-2E\right)^l
\end{equation}
for the density of local minima at any given energy $0\leq E\leq\frac{1}{2}$, in units of the mean eigenvalue of $\bm{H}$ (shifted such that the global minimum is at $E=0$). Representative plots of this distribution in various \change{parameterization} regimes are shown in Figure~\ref{fig:whrf_minima_dist}.

\change{This distribution of local minima is calculated from the joint distribution of the WHRF function value, its gradient, and its Hessian. Thus, by demonstrating the convergence of this joint distribution in the variational loss functions we consider to the analogous distribution in WHRFs at a sufficient rate, we are able to show the same phase transition occurs in variational loss functions. Our full proof is given in the Supplemental Information.}

\subsubsection*{\change{Data Availability}}

\change{The data sets generated during and analysed during the current study are available in the Git repository at \url{https://github.com/bkiani/Beyond-Barren-Plateaus}, under the ``data'' folder.}

\subsubsection*{\change{Code Availability}}

\change{The code used during the current study is available in the Git repository at \url{https://github.com/bkiani/Beyond-Barren-Plateaus}.}

\bibliography{main}

\subsubsection*{Acknowledgments}

The authors thank Giacomo De Palma, Tongyang Li, Seth Lloyd, Milad Marvian, Quynh T. Nguyen, and Agnes Villanyi for helpful feedback and discussions. E.R.A. is supported by the National Science Foundation Graduate Research Fellowship Program under Grant No. 4000063445.

\subsubsection*{Author Contributions}

E.R.A. and B.T.K. both wrote this manuscript and formulated the original project ideas, contributed to the proofs, and performed numerical experiments.

\subsubsection*{Competing Interests}

The authors declare no competing interests.

\end{document}


\maketitle

\section{Training Error Dominates in the Optimization of Variational Quantum Algorithms}\label{app:training_error_dominates}

VQE is purely a problem of optimization and may appear unrelated to the challenges in learning via variational algorithms; however, by decomposing the error of a learning algorithm into key terms using well-established methods \cite{bottou201113}, we will show that variational learning algorithms essentially face the same optimization task and its associated challenges. In both cases, the hardness of learning or optimizing with variational circuits manifests itself in the challenges of optimizing over a cost landscape riddled with traps (or other barriers to optimization).

We restrict ourselves here to the supervised learning framework of empirical risk minimization, where our goal is to learn a space of input and output pairs $(\vx, \vy) \in \mathcal{X} \times \mathcal{Y}$ drawn from a distribution $P(\vx,\vy)$. Given a loss function $\ell: \mathcal{Y} \times \mathcal{Y} \to [0,\infty)$, we quantify how well our function performs by considering the expected risk $\mathcal{R}$:
\begin{equation}
    \mathcal{R}(f) = \mathbb{E}_{\vx}\left[\;\ell(f(\vx),f^*(\vx))\;\right],
\end{equation}
where the expectation above is taken with respect to $P(\vx, \vy)$. To benchmark performance, we compare to the ``optimal" or target function $f^*$ which is the minimizer of the risk:
\begin{equation}
    f^*(\vx) = \argmin_{\hat{\vy} \in \mathcal{Y}} \mathbb{E}\left[\ell(\vy,\hat{\vy})|\vx\right].
\end{equation}
To perform learning, we search for a function $\widehat{f} \in \mathcal{F}$ in the function class $\mathcal{F}$ (think \textit{e.g.,} the set of functions expressed by quantum neural networks). The expected risk $\mathcal{R}(f)$ is not something one can calculate since it requires access to the full probability distribution of the data. Instead, one minimizes the empirical risk $\widehat{\mathcal{R}}(f)$ (often named the training error) over a given training data set $\mathcal{D}$ of size $N$ consisting of pairs $\{\vx_i,\vy_i\}_{i=1}^N$:
\begin{equation}
    \widehat{\mathcal{R}}(f) = \sum_{i=1}^N \ell(f(\vx_i),f^*(\vx_i)).
\end{equation}

Note that we use the hat in $\widehat{\mathcal{R}}$ and $\widehat{f}$ to denote the expected risk measure and function that one actually has access to during training or optimization. Given the above, one can bound the expected risk of any function $\widehat{f}$ as a decomposition below \cite{bottou201113}:
\begin{equation}
\label{eq:error_decomp}
    \mathbb{E}\left[\mathcal{R}(\widehat{f}) - \mathcal{R}(f^*) \right]  \leq
    \underbrace{  \min_{f \in \mathcal{F}} \mathcal{R}(f) -  \mathcal{R}(f^*) }_{\text{approximation error}} + \underbrace{  2 \mathbb{E} \left[ \sup_{f \in \mathcal{F}} \left| \mathcal{R}(f) - \widehat{\mathcal{R}}(f) \right| \right]}_{\text{generalization error}} + \underbrace{ \mathbb{E}\left[ \widehat{\mathcal{R}}(\widehat{f}) - \min_{f \in \mathcal{F}} \widehat{\mathcal{R}}(f) \right] }_{\text{optimization error}},
\end{equation}
where the expectation above is taken with respect to the distribution over data sets or training sets. The proof of this statement follows by a careful, yet straightforward, application of additions/subtractions with corresponding bounds~\cite{bottou201113}.
\begin{proof}
Let $\widehat{f}_{\mathcal{F}} = \argmin_{f \in \mathcal{F}} \widehat{\mathcal{R}}(f)$ and $f_{\mathcal{F}} = \argmin_{f \in \mathcal{F}} \mathcal{R}(f)$. Then, by adding and subtracting quantities, we obtain the following result:
\begin{equation} \label{eq:proof_error_decomp_main}
\begin{split}
    \mathbb{E}\left[\mathcal{R}(\widehat{f}) - \mathcal{R}(f^*) \right] = \mathbb{E}\bigl[ & \mathcal{R}(\widehat{f})  - \mathcal{R}(f^*) \\
    &+ \widehat{\mathcal{R}}(\widehat{f}_{\mathcal{F}}) - \widehat{\mathcal{R}}(\widehat{f}_{\mathcal{F}}) \\
    &+ \mathcal{R}(f_{\mathcal{F}}) - \mathcal{R}(f_{\mathcal{F}}) + \widehat{\mathcal{R}}(f_{\mathcal{F}}) - \widehat{\mathcal{R}}(f_{\mathcal{F}}) \\
    &+ \widehat{\mathcal{R}}(\widehat{f}) - \widehat{\mathcal{R}}(\widehat{f}) \bigr] .
\end{split}
\end{equation}
We reorder the above as follows and note their relation to the main statement:
\begin{align}
    \mathbb{E}\left[\mathcal{R}(\widehat{f}) - \mathcal{R}(f^*) \right] = \; \; &\mathbb{E} \Bigl[  \mathcal{R}(f_{\mathcal{F}})  - \mathcal{R}(f^*) \Bigr] && \text{approximation error} \\
    +&\mathbb{E} \Bigl[  \mathcal{R}(\widehat{f}) - \widehat{\mathcal{R}}(\widehat{f}) \Bigr] && \text{generalization error} \\
    +& \mathbb{E} \Bigl[  \widehat{\mathcal{R}}(f_{\mathcal{F}}) - \mathcal{R}(f_{\mathcal{F}}) \Bigr] && \text{generalization error}   \\
    +& \mathbb{E} \Bigl[  \widehat{\mathcal{R}}(\widehat{f}_{\mathcal{F}}) - \widehat{\mathcal{R}}(f_{\mathcal{F}}) \Bigr] && \leq 0 \text{ since } \widehat{f}_{\mathcal{F}} \text{ minimizes }  \widehat{\mathcal{R}} \\
    +& \mathbb{E} \Bigl[  \widehat{\mathcal{R}}(\widehat{f}) - \widehat{\mathcal{R}}(\widehat{f}_{\mathcal{F}}) \Bigr]  && \text{optimization error}
\end{align}
For the quantities in the generalization error, we have since $\widehat{f}, f_{\mathcal{F}} \in \mathcal{F}$: 
\begin{equation}
    \begin{split}
        \mathbb{E} \Bigl[  \mathcal{R}(\widehat{f}) - \widehat{\mathcal{R}}(\widehat{f}) \Bigr] &\leq \mathbb{E} \left[ \sup_{f \in \mathcal{F}}\left| \mathcal{R}(f) - \widehat{\mathcal{R}}(f) \right| \right] \\
        \mathbb{E} \Bigl[  \widehat{\mathcal{R}}(f_{\mathcal{F}}) - \mathcal{R}(f_{\mathcal{F}}) \Bigr] &\leq \mathbb{E} \left[ \sup_{f \in \mathcal{F}}\left| \mathcal{R}(f) - \widehat{\mathcal{R}}(f) \right| \right] .
    \end{split}
\end{equation}

Plugging these into \eqref{eq:proof_error_decomp_main} and noting as before that $\mathbb{E} \Bigl[  \widehat{\mathcal{R}}(\widehat{f}_{\mathcal{F}}) - \widehat{\mathcal{R}}(f_{\mathcal{F}}) \Bigr] \leq 0$, we arrive at the desired result.

\end{proof}

In the context of quantum variational algorithms, each of these has the following properties:
\begin{itemize}
    \item The \textbf{approximation error} quantifies how well the most optimal function in the hypothesis class $\mathcal{F}$ can fit the function. In variational settings, the approximation error is typically bounded by assuming the target function is generated from a nice class of functions (e.g. shallow circuits) or arguing either analytically or theoretically that a given ansatz can (approximately) express the target function~\cite{nakaji2021expressibility,sim2019expressibility,du2020expressive,shen2020information}.
    \item The \textbf{generalization error} quantifies the statistical error that arises from having a finite data set and is typically insignificant in quantum variational algorithms where circuit complexity is limited with regards to the number of training samples. More precisely, for data sets of size $m$, previous work \cite{caro2021generalization,du2022efficient} bound the generalization error as $\tilde{O}(\sqrt{|G|/m})$ where $|G|$ is the number of trainable gates. In contrast, generalization error in heavily overparameterized classical neural network models are challenging to bound and it is still an open question why deep learning models generalize so well~\cite{zhang2021understanding,neyshabur2017exploring}.
    \item The \textbf{optimization error} measures how well one is able to reduce the empirical risk. Issues with optimization such as poor local minima and barren plateaus arise here. Note that there is a distinct difference between quantum and classical deep learning here. With classical deep neural networks, this quantity is typically negligible since neural networks are overparameterized with respect to the data set size and can fit random data arbitrarily well \cite{zhang2021understanding,maennel2020neural}. Furthermore, due to efficient means of calculating gradients with bit-level precision, classical machine learning algorithms perform optimization over parameters far more efficiently than quantum variational algorithms. In quantum variational models, overparameterization with respect to the Hilbert space dimension is generally needed to arbitrarily fit data \cite{anschuetz2022critical,wiersema2020exploring,kiani2020learning}. Since the Hilbert space dimension grows exponentially with the number of qubits, such overparameterization becomes prohibitive rather rapidly.
\end{itemize}

In summary, the approximation error and generalization error can be bounded efficiently with sufficient data so failures in learning are typically related to optimization over the empirical risk. As an aside, this is loosely analogous to the classical setting of learning polynomial size Boolean circuits which is strongly conjectured to be hard since the space of Boolean functions is challenging to search over \cite{applebaum2008basing}. 

Finally, we would like to stress that the decomposition of the excess risk performed in this section is neither unique nor necessarily tight. The decomposition can be performed in various other ways depending on the quantities one would like to bound. We chose the decomposition here to relate errors in quantum machine learning algorithms to their classical counterparts and to highlight the challenges one may face when attempting to \emph{provably} learn a target function class. 

\section{Statistical Query Framework: Background and Additional Details}
\label{app:SQ_review}

The statistical query (SQ) framework was introduced nearly 25 years ago to analyze the hardness of learning problems~\cite{kearns1998efficient}. This framework restricts algorithms to a series of noisy queries, and hardness results are stated in terms of the number of queries needed to learn a given class of functions. Since there are various different ways of defining the statistical query model---including a recent quantum oracular version proposed in \cite{arunachalam2020quantum}---let us first review some of the various models considered in prior work.
\begin{enumerate}
    \item \textbf{Classical statistical query model}: Introduced by \cite{kearns1998efficient}, this was the first statistical query model introduced. For a given distribution $D$ of inputs over an input space $X$ and target concept $c:X \to \{-1,+1\}$, one can make a statistical query $\operatorname{SQ}(q,\tau)$, by providing a threshold $\tau \in \mathbb{R}^+$ and a query function $q: X \times \{-1,+1\} \to \{-1,+1\}$. The query returns a value in the range:
    \begin{equation}
        \mathbb{E}_{x \sim D}\left[q(x,c(x)) -\tau \right] \leq \operatorname{SQ}(q,\tau) \leq \mathbb{E}_{x \sim D}\left[q(x,c(x)) +\tau \right].
    \end{equation}
    \item \textbf{Correlational statistical query model}: The query is the same as before, except now, one queries correlations $\operatorname{CSQ}(q,\tau)$ only by providing a threshold $\tau \in \mathbb{R}^+$ and a query function $h: X \to \{-1,+1\}$. The query returns a value in the range:
    \begin{equation}
        \mathbb{E}_{x \sim D}\left[h(x)c(x) -\tau \right] \leq \operatorname{CSQ}(h,\tau) \leq \mathbb{E}_{x \sim D}\left[h(x)c(x) +\tau \right].
    \end{equation}
    This model is strictly less powerful than the standard statistical query model since one can perform a correlational statistical query with a standard statistical query \cite{szorenyi2009characterizing}.
    \item \textbf{Quantum statistical query model}: This is a statistical query model with quantum samples \cite{arunachalam2020quantum}. Here, we are restricted to target (classical) Boolean functions $c:\{0,1\}^n\to\{-1,+1\}$. A quantum statistical query $\operatorname{Qstat}(\tau, M)$ is provided with a threshold $\tau \in \mathbb{R}^+$ and an observable or Hamiltonian $M \in (\mathbb{C}^2)^{n+1} \times (\mathbb{C}^2)^{n+1}$ satisfying $\|M\|\leq 1$ and returns a number in the range:
    \begin{equation}
         \bra{\psi_c}M\ket{\psi_c} - \tau \leq \operatorname{Qstat}(M,\tau) \leq \bra{\psi_c}M\ket{\psi_c} + \tau,
    \end{equation}
    where $\ket{\psi_c} = \sum_{x \in \{0,1\}^n} \sqrt{D(x)} \ket{x}\ket{c(x)} $. This model is useful to analyze the hardness of learning classical Boolean functions when given the extra power of querying the classical function in superposition. Our work considers learning quantum data and thus does not fit into the framework of this SQ model.
\end{enumerate}

The SQ learning setting is related to the probably approximately correct (PAC) setting of learning theory \cite{valiant1984theory} in that if an algorithm can learn a given function class in the SQ learning setting under any input distribution, then that function class is also PAC learnable \cite{kearns1998efficient,reyzin2020statistical}. Two very recent works have studied the SQ hardness of learning data generated by quantum circuits. First, \cite{hinsche2021learnability} analyze the hardness of learning the output distribution of clifford circuits and stabilizer states showing that these distributions are hard to learn using classical Boolean SQ oracles. Nevertheless, when given samples from the Boolean hypercube of the distribution, they provide an efficient algorithm based on linear regression to determine the stabilizer state underlying the distribution. Such a result is similar to classic results in \cite{kearns1998efficient} showing that parity functions are hard to learn using only SQ oracle calls but easy when performing linear regression with enough samples. Second, \cite{gollakota2022hardness} show that learning stabilizer states is hard in an SQ setting where queries are made over two-outcome POVMs. Their results show that learning stabilizer states in such a setting is as hard as learning the function class of parity with noise in the standard Boolean setting. Our results expand the set of quantum functions that are hard to learn in SQ settings and relate such hardness results to the variational setting.  

\subsection{Quantum Statistical Query Models}
Variational quantum algorithms are inherently noisy due to unavoidable sources such as the need for sampling outputs, or potentially correctable sources such as gate errors and state preparation noise. In such noisy settings, the statistical query (SQ) model provides a useful framework for quantifying the complexity of learning a class of functions by considering how many query calls to a noisy oracle are needed to learn any function in that class. As described in the main text, in the variational setting, we consider two forms of statistical queries which relate to learning a target Hamiltonian or a target unitary, both of which result in exponential hardness results for learning simple variational classes of data. 

Our proofs expand on recent research showing hardness results in the SQ setting for certain quantum machine learning problems. More specifically, recent results that have shown that certain fundamental and rather simple classes of quantum ``functions'' are hard to learn in the SQ setting. Namely, (classical) output distributions of locally constructed quantum states \cite{hinsche2021learnability} and the set of Clifford circuits \cite{gollakota2022hardness} are hard to learn given properly chosen statistical query oracles. Following these results, we show that simple classes of functions generated by variational circuits are also exponentially difficult to learn in the SQ settings we consider. We also directly connect the statistical query setting to actual optimization algorithms that are used in practice for variational optimization. Our results indicate that training algorithms must be carefully constructed to avoid these poor lower bounds.

\subsection{Limitations of Hardness Results in the SQ Framework}
Though the SQ framework is a useful tool for analyzing the hardness of learning a class of functions in noisy settings, there are a few caveats and limitations of any hardness results proven in the SQ setting:
\begin{itemize}
    \item The statistical query model inherently requires noise in the form of the tolerance $\tau$. Furthermore, the guarantees of learning must handle worst case noise scenarios where the noise acts adversarially on the statistical query. Though quantum variational algorithms are inherently noisy, this noise typically does not arise in an adversarial nature.
    \item The statistical query model places bounds on learning classes of functions using optimizers that query this SQ model and is not directly related to issues of loss landscapes since there is no loss landscape to actually optimize. Nevertheless, since (noisy) calculations of gradients and loss function values are themselves examples of statistical queries, any issues with optimizing over a loss landscape will also arise in performing the optimizer through a series of statistical queries.
    \item Learning every function in a class $\mathcal{C}$ can be restrictive, and in practice, one may only really want to learn a given function or a small set of functions. In fact, it can be shown that even the class of functions generated by shallow neural networks is hard to learn in the SQ setting \cite{diakonikolas2020algorithms,goel2020superpolynomial,goel2020statistical,chen2022hardness,diakonikolas2020near}; nevertheless, neural networks are very successful at learning specific functions such as the classification of real-world images \cite{lecun2015deep}.
    \item Specific to the settings considered here, our hardness results were obtained in the correlational SQ setting by constructing a family of orthogonal functions drawn from a given function class. We chose this setting for its close relation to the algorithms used in practice for performing optimization over variational parameters. However, as mentioned in the main text, the correlational SQ setting is strictly weaker than the more general SQ setting, and separations between SQ and correlational SQ results have been made in prior work~\cite{chen2022learning,andoni2014learning}.
\end{itemize}

\section{Proofs of Statistical Query Results}\label{app:sq_learning_proofs}

Throughout this section, we make use of standard formulas from Weingarten calculus to integrate over Haar measure or $t$-designs \cite{napp2022quantifying,collins2003moments,collins2006integration}. Let $\ket{I_{m}^n}$ denote $n$ copies of the unnormalized maximally mixed state on a Hilbert space of dimension $m$:
\begin{equation}
\label{eq:max_entangle_state}
    \ket{I_m^n} = \sum_{i_1, i_2, \dots, i_n=1}^m \ket{i_1, i_2, \dots, i_n}\ket{i_1, i_2, \dots, i_n}.
\end{equation}
For $n=2$, let $\ket{S_m^2}$ denote the same unnormalized state as above with a swap operation applied to the second register:
\begin{equation}
\label{eq:max_entangle_state_swap}
\begin{split}
    \ket{S_m^2} &= \left( \mI \otimes \operatorname{SWAP} \right) \ket{I_m^2} \\
    &= \sum_{i_1, i_2 = 1}^m \ket{i_1, i_2} \ket{i_2, i_1}.
\end{split}
\end{equation}

The following hold over a distribution $\mathcal{D}$ that is a $2$-design over the unitary matrices of dimension $m$:
\begin{equation}
\label{eq:haar_moments}
    \begin{split}
        \mathbb{E}_{\mU \sim \mathcal{D}}\left[\mU \otimes \bar{\mU} \right] &= \frac{1}{m}  \ket{I_m^1} \bra{I_m^1} , \\
        \mathbb{E}_{\mU \sim \mathcal{D}}\left[\mU \otimes \mU \otimes \bar{\mU} \otimes \bar{\mU} \right] &= \frac{1}{m^2-1} \left( \ket{I_m^2}  \bra{I_m^2} + \ket{S_m^2}  \bra{S_m^2}  \right) - \frac{1}{m(m^2-1)} \left( \ket{I_m^2} \bra{S_m^2} + \ket{S_m^2} \bra{I_m^2} \right),
    \end{split}
\end{equation}
where $\bar{\mU}$ denotes the matrix with entries that are the complex conjugate of entries of $\mU$.

As a simple example of applying the techniques above, we show that for unitaries $\mU_*$ and $\mV$ of dimension $d^n$ (e.g., $d=2$ for qubits and $n$ is the number of qubits), $\mathbb{E}_{\rho \sim \mathcal{D}}\left[ \Re[ \Tr( \mU_*^\dagger \mV \rho ) ] \right] = d^{-n} \Re[ \Tr(\mU_*^\dagger \mV) ]$ whenever $\mathcal{D}$ forms a $1$-design. This is a crucial formula that we use in the evaluation of statistical queries to $\operatorname{qUSQ}$.
\begin{lemma} \label{lem:ave_qusq_2design}
For any distribution $\mathcal{D}$ that is a $1$-design over states of dimension $d^n$,
\begin{equation}
    \mathbb{E}_{\rho \sim \mathcal{D}}\left[ \Re[ \Tr( \mW^\dagger \mV \rho ) ] \right] = d^{-n} \Re[ \Tr(\mW^\dagger \mV) ].
\end{equation}
\end{lemma}
\begin{proof}
WLOG, we rewrite the equation above in terms of a distribution over pure states and with a slight abuse of notation, we let $\mathcal{D}$ also denote a distribution over unitary matrices $\mU$ that forms a $1$-design:
\begin{equation}
\begin{split}
    \mathbb{E}_{\rho \sim \mathcal{D}}\left[ \Re[ \Tr( \mW^\dagger \mV \rho ) ] \right] &= \mathbb{E}_{\mU \sim \mathcal{D}}\left[ \Re[ \bra{0} \mU^\dagger \mW^\dagger \mV \mU \ket{0} ] \right].
\end{split}    
\end{equation}

Using \eqref{eq:max_entangle_state}, we have:
\begin{equation}
\begin{split}
    \mathbb{E}_{\mU \sim \mathcal{D}}\left[ \Re[ \bra{0} \mU^\dagger \mW^\dagger \mV \mU \ket{0} ] \right]
    &=\mathbb{E}_{\mU \sim \mathcal{D}} \left[  \bra{I_{d^n}^1}\bigl( ( \mW^\dagger \mV) \otimes \mI \bigr) \bigl(\mU \otimes \bar{\mU} \bigr) \ket{0} \ket{0} \right].
\end{split}    
\end{equation}

Applying \eqref{eq:haar_moments}, we have:
\begin{equation}
\begin{split}
    \mathbb{E}_{\mU \sim \mathcal{D}}\left[ \Re[ \bra{0} \mU^\dagger \mW^\dagger \mV \mU \ket{0} ] \right]
    &=\mathbb{E}_{\mU \sim \mathcal{D}} \left[ \Re[  \bra{I_{d^n}^1}\bigl( ( \mW^\dagger \mV) \otimes \mI \bigr) \bigl(\mU \otimes \bar{\mU} \bigr) \ket{0} \ket{0} ]\right]\\
    &= \frac{1}{d^n} \Re[ \bra{I_{d^n}^1}\bigl( ( \mW^\dagger \mV) \otimes \mI \bigr) \ket{I_{d^n}^1} \bra{I_{d^n}^1} \ket{0} \ket{0} ] \\
    &= \frac{1}{d^n} \Re[ \Tr(\mW^\dagger \mV) ] .
\end{split}    
\end{equation}
\end{proof}

\subsection{Lower Bounds for Statistical Query Learning}
\label{app:sq_lower_bound_classical_proof}
For completeness, we include a proof of \Cref{thm:sq_lower_bound}, copied below, which lower bounds the number of queries needed to learn a hypothesis class. $\operatorname{qCSQ}$ and $\operatorname{qUSQ}$ both take the form of an inner product so the proof holds for both cases. As a reminder we include the definitions of $\operatorname{qCSQ}$ and $\operatorname{qUSQ}$ below.

\paragraph{Quantum correlational statistical query (qCSQ)} Given a target observable $\mM$ and a distribution $\mathcal{D}$ of input states, a query $\operatorname{qCSQ}(\mO,\tau)$ takes in a bounded observable $\mO$ with $\| \mO \| \leq 1$ and a tolerance $\tau$ and returns a value in the range:
\begin{equation}
    \mathbb{E}_{\rho \sim \mathcal{D}}\left[\operatorname{Tr}(\mO\rho) \operatorname{Tr}(\mM\rho) -\tau \right] \leq \operatorname{qCSQ}(\mO,\tau) \leq \mathbb{E}_{\rho \sim \mathcal{D}}\left[\operatorname{Tr}(\mO\rho) \operatorname{Tr}(\mM\rho) +\tau \right].
\end{equation}

\paragraph{Quantum unitary statistical query (qUSQ)} Given a target unitary transformation $\mU_*$ over a distribution $\mathcal{D}$ of inputs, the oracle $\operatorname{qUSQ}(\mV, \tau)$ takes in a unitary matrix $\mV$ and a tolerance $\tau$ and returns a value in the range:
\begin{equation}
    \mathbb{E}_{\rho \sim \mathcal{D}}\left[\Re[\operatorname{Tr}(\mU_*^\dagger \mV \rho)] -\tau \right] \leq \operatorname{qUSQ}(\mV,\tau) \leq \mathbb{E}_{\rho \sim \mathcal{D}}\left[\Re[\operatorname{Tr}(\mU_*^\dagger \mV \rho)] +\tau \right].
\end{equation}

Finally, we remind the reader of the definition of the statistical query dimension.
\begin{restatable}[Statistical query dimension \cite{blum1994weakly,reyzin2020statistical}]{definition}{SQdimDefinition}
For a distribution $\mathcal{D}$ and concept class $\mathcal{H}$ where $\|\mM\|_{\mathcal{D}}^2\leq C_{max}$ for all $\mM \in \mathcal{H}$, the statistical query dimension ($\operatorname{SQ-DIM}_{\mathcal{D}}(\mathcal{H})$) is the largest positive integer $d$ such that there exists $d$ observables $\mM_1, \mM_2, \dots, \mM_d \in \mathcal{H}$ such that for all $i \neq j:$ $| \langle \mM_i, \mM_j \rangle_{\mathcal{D}} | \leq C_{max}/d$.
\end{restatable}

From here, we provide a proof of the query complexity of learning a function class, similar to the proof in \cite{szorenyi2009characterizing}.

\begin{restatable}[Query complexity of learning~\cite{szorenyi2009characterizing,blum1994weakly}]{theorem}{SQlowerbound}
\label{thm:sq_lower_bound}
Given a distribution $\mathcal{D}$ on inputs and a hypothesis class $\mathcal{H}$ where $\| \mM \|_{\mathcal{D}}^2 \leq C_{max}$ for all $\mM \in \mathcal{H}$, let $d=\operatorname{SQ-DIM}_{\mathcal{D}}(\mathcal{H})$ be the statistical query dimension of $\mathcal{H}$. Any $\operatorname{qCSQ}$ or $\operatorname{qUSQ}$ learner making queries with tolerance $C_{max}\tau$ must make at least $(d \tau^2 -1 )/2 $ queries to learn $\mathcal{H}$ up to error $C_{max}\tau$.
\end{restatable}

\begin{proof}
Since we are restricted to the weaker setting of correlational statistical queries in this study, we can reuse a simple and elegant proof from \cite{szorenyi2009characterizing}.

Let $\mM_1, \mM_2, \cdots, \mM_d$ be $d$ functions that saturate $\operatorname{SQ-DIM}_{\mathcal{D}}(\mathcal{H})$, i.e., $\langle \mM_i, \mM_j \rangle_{\mathcal{D}} \leq 1/d$ for all $i \neq j$. Assume we apply query $\mO$ and let $S = \{i \in [d]: \langle \mO, \mM_i \rangle_{\mathcal{D}} > C_{max} \tau \}$ Then, by simple application of Cauchy-Schwarz, we have that for any query $\mO$:
\begin{equation}
    \begin{split}
        \left\langle \mO,  \sum_{i \in S} \mM_i \right\rangle_{\mathcal{D}}^2 & \leq C_{max}  \left\| \sum_{i \in S} \mM_i  \right\|_{\mathcal{D}}^2 \\
        &= C_{max} \sum_{i,j \in S} \langle \mM_i, \mM_j \rangle_{\mathcal{D}} \\
        & \leq C_{max}^2 \left( |S| + |S|^2 / d \right).
    \end{split}
\end{equation}
Note, that we can also bound the quantity above from below by using the definition of $S$:
\begin{equation}
    \left\langle \mO,  \sum_{i \in S} \mM_i \right\rangle_{\mathcal{D}} \geq C_{max} |S|\tau.
\end{equation}

Combining the above, we have that
\begin{equation}
    |S| \leq d / (d C_{max}^2 \tau^2 - 1 ).
\end{equation}

Similarly, defining $S'= \{i \in [d]: \langle \mO, \mM_i \rangle_{\mathcal{D}} < -C_{max}\tau \}$ with correlation less than $-\tau$, we follow the steps above to also note that $|S'| \leq d / (d C_{max}^2 \tau^2 - 1 )$. Altogether, we have that $|S'| + |S| \leq 2d / (d C_{max}^2 \tau^2 - 1 )$, which implies that each oracle call can only eliminate up to $2d / (d C_{max}^2 \tau^2 - 1 )$ functions. Since we must eliminate at least $d$ functions to learn the target class, we arrive at the desired bound. 
\end{proof}

\subsection{Proofs of Statistical Query Dimensions for Variational Function Classes}

\begin{proposition}[SQ dimension for $L=1$ and fixed global measurement]
\label{prop:sq_dim_l1_fixedglobal}
Given $n$ qubits, let $\mathcal{H}$ be the concept class containing functions $f:\mathbb{C}^{2^n} \to \mathbb{R}$ consisting of single qubit rotations followed by a global Pauli Z measurement, i.e. functions of the form
\begin{equation}
    f(\ket{\psi};  \mU_1, \mU_2, \dots, \mU_n) = \bra{\psi} \left( \mU_1^\dagger \otimes \mU_2^\dagger \otimes \cdots \otimes \mU_n^\dagger \right) \left( \mZ_1 \otimes \mZ_2 \otimes \cdots \mZ_n \right) \left( \mU_1 \otimes \mU_2 \otimes \cdots \otimes \mU_n \right) \ket{\psi},
\end{equation}
where $\ket{\psi}$ is the input to the function and $\mU_1, \mU_2, \dots, \mU_n$ are the parameterized $1$-qubit rotation operations on distinct qubits. Then, the concept class $\mathcal{H}$ has SQ dimension $\operatorname{SQ-DIM}_{\mathcal{D}}(\mathcal{H}) \geq 3^n$ under any distribution of states that forms a $2$-design.
\end{proposition}
\begin{proof}
The simple proof of this proposition relies on the fact that all Pauli operators are pairwise orthogonal for a $2$-design, i.e. given two distinct Pauli operators $\mP_1$ and $\mP_2$ then $\mathbb{E}_{\rho \sim \mathcal{D}}\left[\operatorname{Tr}(\mP_1\rho) \operatorname{Tr}(\mP_2\rho) \right] = 0$. Therefore, we simply show that the concept class $\mathcal{H}$ is capable of producing any Pauli string not containing the identity. 

To proceed, note that we can rewrite the function class as follows:
\begin{equation}
    f(\ket{\psi}; \mU_1, \mU_2, \dots, \mU_n) = \bra{\psi} (\mU_1^\dagger \mZ_1 \mU_1) \otimes (\mU_2^\dagger \mZ_2 \mU_2) \otimes \cdots \otimes ( \mU_n^\dagger \mZ_n \mU_n)  \ket{\psi}.
\end{equation}

To obtain any arbitrary Pauli string, we simply conjugate the $\mZ_i$ operator for the $i$-th qubit by a corresponding operation. If the $i$-th qubit of a Pauli string is equal to $\mX$, then we set $\mU_i=\mH$ or the Hadamard transform. Similarly, if the $i$-th qubit of a Pauli string is equal to $\mY$, then we set $\mU_i=\mH \sqrt{\mZ}^\dagger$. By conjugation of the individual 1-qubit operators, we thus can produce any Pauli operator in $\{ \mX, \mY, \mZ \}^{\otimes n}$.
\end{proof}

\begin{corollary} \label{cor:sq_lb_l1_fixedglobal}
By application of \Cref{thm:sq_lower_bound}, the class of functions defined in \Cref{prop:sq_dim_l1_fixedglobal} consisting of a single layer of parameterized single qubit unitary gates and a fixed global measurement on $n$ qubits requires $2^{\operatorname{\Omega}(n)}$ queries to learn for a query tolerance greater than $3^{-\beta n}$, where $\beta = 1/2 - \operatorname{\Omega}(1)$.
\end{corollary}

\begin{proposition}[SQ dimension for $L=\ceil{\log_2 n}$, two-qubit gates, and single Pauli $\mZ$ measurement]
\label{prop:sq_lb_logn_fixedsingle}
Given $n$ qubits, let $\mathcal{H}$ be the concept class containing functions $f:\mathbb{C}^{2^n} \to \mathbb{R}$ consisting of $\ceil{\log_2 n}$ layers of two-qubit gates followed by a Pauli $\mZ$ measurement on a single qubit. Then, the concept class $\mathcal{H}$ has SQ dimension $\operatorname{SQ-DIM}_{\mathcal{D}}(\mathcal{H}) \geq 4^n-1$ under any distribution of inputs that forms a $2$-design.
\end{proposition}
\begin{proof}
We will show that $\mathcal{H}$ is powerful enough to perform any nontrivial Pauli measurement (i.e., any Pauli but the identity) and hence construct at least $4^n-1$ orthogonal functions. Classically, any parity function can be constructed in $\ceil{\log_2 n}$ layers, and we use a similar construction here. 

Without loss of generality, assume the Pauli measurement is on the first qubit. Let $\mU(\theta)$ represent a possible unitary that can be applied using the given hypothesis class, resulting in a final measurement of $\mU(\theta)^\dagger \mZ_1 \mU(\theta)$ on a given input state $\ket{\psi}$. We will show that we can parameterize the circuit such that for any Pauli measurement, $\mP_1 \otimes \mP_2 \otimes \cdots \otimes \mP_n = \mU(\theta)^\dagger \mZ_1 \mU(\theta)$ where $\mP_i$ indicates the Pauli operator of qubit $i$ (i.e., $\mP_i \in \{ \mI, \mX, \mY, \mZ \}$). 

To construct any Pauli operator $\mP_1 \otimes \mP_2 \otimes \cdots \otimes \mP_n$, we follow the steps below:
\begin{enumerate}
    \item In the first layer, apply a unitary to each qubit $i$ which maps the computational basis to the basis of the Pauli for qubit $i$. In more detail, if $\mP_i = \mI$ or $\mP_i = \mZ$, then apply the identity map to keep the basis the same. If $\mP_i = \mX$, then apply the Hadamard transform and if $\mP_i = \mY$ then apply the operation $\mH \sqrt{\mZ}^\dagger$.  
    \item In the $l$-th layer, apply a specific two qubit gate to qubit pairs $\{1, 2^{l-1}+1\}, \{2(2^{l-1})+1, 3(2^{l-1})+1\}, \{4(2^{l-1})+1, 5(2^{l-1})+1\}, \dots$. For a layer $l$ and a given pair $\{i,j\}$, apply the following gate:
    \begin{itemize}
        \item if all of $\mP_i, \mP_{i+1}, \dots, \mP_{j+2^{l-1}}$ are equal to $\mI$, then apply the identity.
        \item if any of $\mP_i, \mP_{i+1}, \dots, \mP_{j-1}$ are not equal to $\mI$ and all of  $\mP_j, \mP_{j+1}, \dots, \mP_{j+2^{l-1}}$ are equal to $\mI$ then apply the identity as well.
        \item if all of $\mP_i, \mP_{i+1}, \dots, \mP_{j-1}$ are equal to $\mI$ and any of $\mP_j, \mP_{j+1}, \dots, \mP_{j+2^{l-1}}$ are not equal to $\mI$, then apply a swap gate between qubits $i$ and $j$.
        \item otherwise, apply the following $2$-qubit gate to $i$ and $j$ which conjugates $\mZ \otimes \mI$ to $\mZ \otimes \mZ$:
        \begin{equation}
            \begin{bmatrix}
                    1 & 0 & 0 & 0 \\ 0 & 0 & 0 & 1 \\ 0 & 0 & 1 & 0 \\ 0 & 1 & 0 & 0
            \end{bmatrix}.
        \end{equation}
    \end{itemize} 
    \item repeat step 2 above starting from $l=1$ to $l=\ceil{\log_2 n}$. Measuring the first qubit will measure the corresponding desired Pauli. Note, that the single qubit operations of step 1 and the two-qubit operations of step 2 can be combined into a single 2-qubit gate thus not changing the depth.
\end{enumerate}

Following the steps above, at layer $l$, the measurement of the first qubit corresponds to the Pauli measurement of the first $2^l$ qubits. Recursively applying this procedure $l$ layers produces any arbitrary Pauli string.
\end{proof}

\begin{corollary} \label{cor:sq_lb_logn_fixedsingle}
By application of \Cref{thm:sq_lower_bound}, the class of functions defined in \Cref{prop:sq_lb_logn_fixedsingle} consisting of $\ceil{\log_2 n}$ two-qubit unitary gates and a fixed measurement on a single qubit requires $2^{\operatorname{\Omega}(n)}$ queries to learn for a query tolerance greater than $4^{-\beta n}$, where $\beta =  1/2 - \operatorname{\Omega}(1)$.
\end{corollary}

\begin{proposition}[SQ dimension for $L$ layers, neighboring $2$-local gates in one-dimensional lattice and fixed single qubit measurement]
\label{prop:sq_dim_l_neighboring_singlequbit}
Given $n$ qubits, let $\mathcal{H}$ be the concept class containing functions $f:\mathbb{C}^{2^n} \to \mathbb{R}$ consisting of $L$ layers of $2$-qubit unitary operations followed by a Pauli Z measurement on a single qubit (labeled qubit $m$), i.e. functions of the form
\begin{equation}
    f(\ket{\psi}; \mW_1, \mW_2, \dots, \mW_L) = \bra{\psi} \mW_{1}^\dagger \mW_2^\dagger  \cdots \mW_L^\dagger \left( \mZ_m \right) \mW_L \cdots \mW_2 \mW_1 \ket{\psi},
\end{equation}
where $\ket{\psi}$ is the input to the function and $\mW_1, \mW_2, \dots, \mW_L$ are the unitary operations at each layer consisting of tensor products of $2$-local unitary operators acting on neighboring qubits. Then, the concept class $\mathcal{H}$ has SQ dimension $\operatorname{SQ-DIM}_{\mathcal{D}}(\mathcal{H}) \geq 4^{\min\left(2L,n\right)}-1$ under any distribution of states that forms a $2$-design.
\end{proposition}
\begin{proof}
Our proof relies on the fact that with $L$ layers, one can conjugate the fixed single qubit measurement on qubit $m$ to produce any $2L$-qubit Pauli on the $2L$ qubits within the reverse light cone of $m$. We follow a proof outline similar to \Cref{prop:sq_lb_logn_fixedsingle}. 

Before we proceed, we assume without loss of generality, that $L$ is odd and the first layer applies a two qubit unitary to qubit $m$ and the preceeding qubit $m-1$. It is straightforward to extend this to the case where $L$ is even. Therefore, qubit $m$ is the $L$-th qubit in the reverse light cone of qubit $m$, i.e., the light cone traverses qubits $m-L$ to $m+L-1$. For the steps below, we then index the qubits from $-L$ to $L-1$ so that the numbering is relative to qubit $m$. To perform a given $2L$ Pauli operator $\mP_{-L} \otimes \mP_{-L+1} \otimes \cdots \otimes \mP_{L-1}$ in the reverse light cone of qubit $m$, we follow the steps below, many of which are copied from \Cref{prop:sq_lb_logn_fixedsingle}.: 
\begin{enumerate}
    \item In the first layer, apply a unitary to each qubit $i$ which maps the computational basis to the basis of the Pauli for qubit $i$. In more detail, if $\mP_i = \mI$ or $\mP_i = \mZ$, then apply the identity map to keep the basis the same. If $\mP_i = \mX$, then apply the Hadamard transform and if $\mP_i = \mY$ then apply the operation $\mH \sqrt{\mZ}^\dagger$.  
    \item in the $L$-th layer, for the $2$-qubit unitary acting on qubits $0$ and $-1$, apply the following gate:
    \begin{itemize}
        \item if all of $\mP_{-L}, \mP_{-L+1}, \dots, \mP_{-1}$ are equal to $\mI$ and all of $\mP_{1}, \mP_{2}, \dots, \mP_{L-1}$ are equal to $\mI$, then apply the identity.
        \item if any of $\mP_{-L}, \mP_{-L+1}, \dots, \mP_{-1}$ are not equal to $\mI$ and and any of $\mP_{1},\mP_{2}, \dots, \mP_{L-1}$  are not equal to $\mI$, apply the following $2$-qubit gate ($\operatorname{CNOT}$) to $i$ and $i+1$ which conjugates $\mI \otimes \mZ$ to $\mZ \otimes \mZ$:
        \begin{equation}
            \begin{bmatrix}
                    1 & 0 & 0 & 0 \\ 0 & 1 & 0 & 0 \\ 0 & 0 & 0 & 1 \\ 0 & 0 & 1 & 0
            \end{bmatrix}.
        \end{equation}
        \item otherwise, apply the swap operation between qubits $0$ and $-1$.
    \end{itemize}
    \item In the $l$-th layer for any $l \neq L$, apply a specific two qubit gate to neighboring qubit pairs $\{-L+l-1,-L+l\}$ on the edge of the reverse light cone. For simplicity, let $i=-L+1-1$ and apply the following gate to qubit pair $\{i,i+1\}$:
    \begin{itemize}
        \item if all of $\mP_{-L}, \mP_{-L+1}, \dots, \mP_{i}$ are equal to $\mI$, then apply the identity.
        \item if any of $\mP_{-L}, \mP_{-L+1}, \dots, \mP_{i}$ are not equal to $\mI$ and $\mP_{i+1}=\mI$, then apply a swap between qubits $i$ and $i+1$.
        \item otherwise, apply the following $2$-qubit gate ($\operatorname{CNOT}$) to $i$ and $i+1$ which conjugates $\mI \otimes \mZ$ to $\mZ \otimes \mZ$:
        \begin{equation}
            \begin{bmatrix}
                    1 & 0 & 0 & 0 \\ 0 & 1 & 0 & 0 \\ 0 & 0 & 0 & 1 \\ 0 & 0 & 1 & 0
            \end{bmatrix}.
        \end{equation}
    \end{itemize} 
    Similarly, for the other edge of the reverse light cone, we apply the same gates, but in ``reverse" logic. Here, we apply a $2$-qubit unitary to qubit pair $\{L-l-2, L-l-1\}$. For simplicity, let $i=L-l-2$ and apply the following gate to qubit pair $\{i,i+1\}$:
    \begin{itemize}
        \item if all of $\mP_{i+1}, \mP_{i+2}, \dots, \mP_{L-1}$ are equal to $\mI$, then apply the identity.
        \item if any of $\mP_{i+1}, \mP_{i+2}, \dots, \mP_{L-1}$ are not equal to $\mI$ and $\mP_{i+1}=\mI$, then apply a swap between qubits $i$ and $i+1$.
        \item otherwise, apply the following $2$-qubit gate to $i$ and $i+1$ which conjugates $\mZ \otimes \mI$ to $\mZ \otimes \mZ$:
        \begin{equation}
            \begin{bmatrix}
                    1 & 0 & 0 & 0 \\ 0 & 0 & 0 & 1 \\ 0 & 0 & 1 & 0 \\ 0 & 1 & 0 & 0
            \end{bmatrix}.
        \end{equation}
    \end{itemize} 
    \item repeat step 3 above starting from $l=1$ to $l=L-1$. Measuring the $m$-th qubit will measure the corresponding desired Pauli. Note, that the single qubit operations of step 1 and the two-qubit operations of step 2 can be combined into a single 2-qubit gate thus not changing the depth.
\end{enumerate}

\end{proof}

\begin{corollary} \label{cor:sq_lb_l_neighboring_singlequbit}
By application of \Cref{thm:sq_lower_bound}, the class of functions defined in \Cref{prop:sq_dim_l_neighboring_singlequbit} consisting of $L$ layers of neighboring $2$-qubit gates and a fixed measurement on a single qubit requires $2^{\Omega\left(\min\left(2L,n\right)\right)}$ queries to learn for a constant query tolerance that does not depend on $L$ and $n$.
\end{corollary}

The above can be generalized to lower bound the statistical query dimension for circuits of $L$ layers on $d$-dimensional lattices as we show below. In $d$-dimensional lattices, since the light cone of a single qubit measurement grows at a rate of $L^d$ for an $L$ layers, we can prove that the statistical query dimension grows as $2^{\Omega\left(\min\left(2L,n^{1/d}\right)^d\right)}$.

\begin{figure*}[!ht]
    \centering
    \includegraphics[]{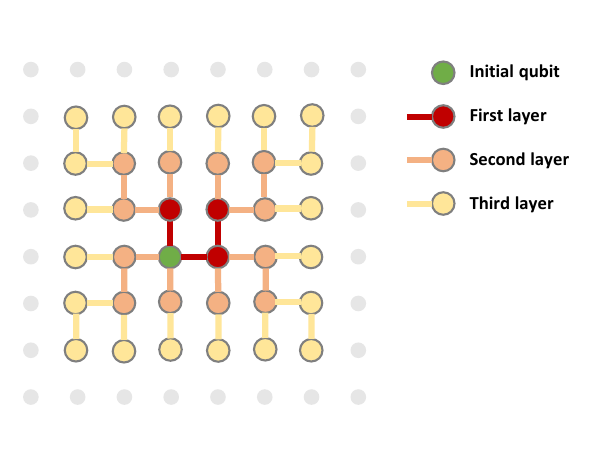}
    \caption{Growth of the light cone for a 2-dimensional lattice, where the initial qubit is the one that is measured. The size of the lattice grows with the perimeter of the light cone for each layer which consists of local $2$-qubit gates applied in each dimension. Each qubit is connected to a qubit in the edge of the light cone of the prior layer, forming a graph which is a tree rooted at the initial qubit.} \label{fig:lattice_appendix}
\end{figure*}

\begin{proposition}[SQ dimension for $L$ layers, neighboring $2$-local gates on $d$-dimensional lattice and fixed single qubit measurement]
\label{prop:sq_dim_l_neighboring_singlequbit_lattice}
Given $n$ qubits, let $\mathcal{H}$ be the concept class containing functions $f:\mathbb{C}^{2^n} \to \mathbb{R}$ consisting of $L$ layers of $2$-qubit unitary operations applied in each dimension followed by a Pauli Z measurement on a single qubit (labeled qubit $m$), i.e. functions of the form
\begin{equation}
    f(\ket{\psi}; \mW_1, \mW_2, \dots, \mW_L) = \bra{\psi} \mW_{1}^\dagger \mW_2^\dagger \cdots  \mW_L^\dagger \left( \mZ_m \right) \mW_L \cdots \mW_2 \mW_1 \ket{\psi},
\end{equation}
where $\ket{\psi}$ is the input to the function and $\mW_1, \mW_2, \dots, \mW_L$ are the unitary operations at each layer consisting of tensor products of $2$-local unitary operators acting along each dimension on neighboring qubits in a $d$-dimensional lattice. Then, the concept class $\mathcal{H}$ has SQ dimension $\operatorname{SQ-DIM}_{\mathcal{D}}(\mathcal{H}) = 2^{\Omega\left(\min\left(2L,n^{1/d}\right)^d\right)}$ under any distribution of states that forms a $2$-design.
\end{proposition}
\begin{proof}
Our proof relies on the fact that with $L$ layers, one can conjugate the fixed single qubit measurement on qubit $m$ to produce any Pauli on the $\operatorname{\Omega}(L^d)$ qubits within the reverse light cone of $m$. We follow a proof outline similar to \Cref{prop:sq_dim_l_neighboring_singlequbit}. 

To be more precise, let us introduce some notation. To perform any Pauli measurement in the reverse light cone at a given layer $l \in [L]$ indexed in reverse order, we apply gates to the perimeter of the reverse light cone at layer $l-1$. We assume there are $N_l$ qubits in the reverse light cone at layer $l$ and index these qubits from $1$ to $N_l$ to construct the Pauli $\mP_1 \otimes \mP_2 \otimes \cdots \otimes \mP_{N_l}$. Like in \Cref{prop:sq_dim_l_neighboring_singlequbit}, we grow the Pauli at each layer. 

To grow the light cone and properly choose the $2$-qubit gates, we construct a graph which is a tree where the parent of any qubit is the prior qubit which it was connected to in the light cone of the previous layer (see \Cref{fig:lattice_appendix} for an example). The root of the tree is the qubit which is being measured. For example, at layer $l=1$, the light cone is of size two in each dimension and the qubit being measured is the parent to the child node which it is connected to. To construct any pauli $\mP_1 \otimes \mP_2 \otimes \cdots \otimes \mP_{N_L}$, we follow the steps below:
\begin{enumerate}
    \item in the $l$-th layer, for all parent and child qubits $p$ and $c$ respectively connected in the tree at layer $l$, apply a unitary acting on qubits $p$ and $c$ as follows:
    \begin{itemize}
        \item if all of the qubits that are descendants of qubit $c$ and qubit $c$ itself have Pauli terms that are equal to $\mI$, then apply the identity gate between qubit $p$ and $c$.
        \item if any of the qubits that are descendants of qubit $c$ or qubit $c$ itself have Pauli terms that are not equal to $\mI$ and the Pauli term of qubit $p$ is equal to $\mI$, then apply the swap gate between $p$ and $c$.
        \item otherwise, apply the following $2$-qubit gate to qubits $p$ and $c$ which conjugates $\mZ \otimes \mI$ to $\mZ \otimes \mZ$:
        \begin{equation}
            \begin{bmatrix}
                    1 & 0 & 0 & 0 \\ 0 & 0 & 0 & 1 \\ 0 & 0 & 1 & 0 \\ 0 & 1 & 0 & 0
            \end{bmatrix}.
        \end{equation}
    \end{itemize} 
    \item repeat the step above starting from $l=1$ to $l=L$. 
    \item In the first layer ($l = L$), apply a unitary to each qubit $i$ which maps the computational basis to the basis of the Pauli for qubit $i$. In more detail, if $\mP_i = \mI$ or $\mP_i = \mZ$, then apply the identity map to keep the basis the same. If $\mP_i = \mX$, then apply the Hadamard transform and if $\mP_i = \mY$ then apply the operation $\mH \sqrt{\mZ}^\dagger$.  
\end{enumerate}
Following the above steps, measuring the single qubit will measure the corresponding desired Pauli. Note, that the single qubit operations of the first layer and the two-qubit operations of that layer can be combined into a single 2-qubit gate thus not changing the depth.

In each layer, $2$-qubit gates act along each dimension in some order. We can assume an ordering of the dimensions without loss of generality and assume that we apply gates along that dimension in order. After the first layer, the lattice has size $2$ along each dimension. For each layer thereafter, the lattice grows by $2$ qubits in each dimension (see \Cref{fig:lattice_appendix} for example). Therefore, the reverse light cone grows at a rate $\operatorname{\Omega}(L^d)$. Since the light cone can be at most of size $n$ (number of qubits), then the light cone is of size $2^{\operatorname{\Omega}\left(\min\left(2L,n^{1/d}\right)^d\right)}$ for all $L$ layers.
\end{proof}

\begin{corollary} \label{cor:sq_lb_l_neighboring_singlequbit_lattice}
By application of \Cref{thm:sq_lower_bound}, the class of functions defined in \Cref{prop:sq_dim_l_neighboring_singlequbit_lattice} consisting of $L$ layers of neighboring $2$-qubit gates and a fixed measurement on a single qubit requires $2^{\operatorname{\Omega}\left(\min\left(2L,n^{1/d}\right)^d\right)}$ queries to learn for a query tolerance that decays no faster than $2^{\operatorname{\upomega}\left(\min\left(2L,n^{1/d}\right)^d\right)}$. For $L \ll n$, this is equal to $2^{\operatorname{\Omega}(L^d)}$ for any constant query tolerance that does not depend on $L$ or $n$.
\end{corollary}

\begin{proposition}[SQ dimension for $L=1$, unitary compiling, and single qubit gates]
\label{prop:sqdim_1layer_unitary_compiling}
Given $n$ qubits, let $\mathcal{H}$ be the concept class containing unitary transformations $\mV:\mathbb{C}^{2^n} \to \mathbb{C}^{2^n}$ consisting of single qubit rotations in a single layer
\begin{equation}
    \mV(\ket{\psi}, \mU_1, \mU_2, \dots, \mU_n) = \mU_1 \otimes \mU_2 \otimes \cdots \otimes \mU_n \ket{\psi},
\end{equation}
where $\ket{\psi}$ is the input to the transformation and $\mU_1, \mU_2, \dots, \mU_n$ are the parameterized $1$-qubit operations. Then, the concept class $\mathcal{H}$ has SQ dimension $\operatorname{SQ-DIM}_{\mathcal{D}}(\mathcal{H}) \geq 4^n$ under the $\operatorname{qUSQ}$ model and any distribution $\mathcal{D}$ of inputs that is a $2$-design.
\end{proposition}
\begin{proof}
From \Cref{lem:ave_qusq_2design}, we have that $\langle \mU, \mV \rangle_{\mathcal{D}} = 2^{-n} \Re\left[ \Tr(\mU^\dagger \mV) \right]$. With one layer of single qubit unitary operations, any Pauli matrix can be constructed. Since $\Tr(\mP_1 \mP_2)=0$ for any two distinct Pauli matrices $\mP_1$ and $\mP_2$, there are at least $4^n$ matrices in $\mathcal{H}$ which are orthogonal under the inner product. 
\end{proof}

\begin{corollary} \label{cor:sq_lb_unitary_single}
By application of \Cref{thm:sq_lower_bound}, the class of functions defined in \Cref{prop:sqdim_1layer_unitary_compiling} consisting of a single layer of single qubit unitaries requires $2^{\operatorname{\Omega}(n)}$ queries to learn for a query tolerance greater than $4^{-\beta n}$, where $\beta <  1/2 - \operatorname{\Omega}(1)$.
\end{corollary}

\subsection{Swap Test via Statistical Queries}

In the task of unitary compiling, one is given copies of states which are inputs and outputs of a target unitary transformation, and the goal is to learn the unitary transformation from those states. More formally, we aim to learn a unitary $\mU_*$ given a distribution over inputs or a dataset of $m$ state pairs $\{ \ket{\phi_i}, \mU_* \ket{\phi_i} \}_{i \in [m]}$. 

One means of measuring overlaps between states is via the swap test \cite{buhrman2001quantum}. For pure states, $\ket{\phi}$ and $\ket{\psi}$, the swap test measures the fidelity $| \braket{\phi}{\psi} |^2$. The measured register in the swap test outputs $\ket{0}$ with probability $1/2+| \braket{\phi}{\psi} |^2/2$ and $\ket{1}$ otherwise. As we show in the main text, this quantity can be calculated using queries to $\operatorname{qUSQ}$. We use the helper lemma below to prove this fact.

\begin{lemma} \label{lem:ave_fidelity_sq_2design}
For any distribution $\mathcal{D}$ that is a $2$-design on a Hilbert space of dimension $m$,
\begin{equation}
    \mathbb{E}_{\rho \sim \mathcal{D}}\left[ \Tr( \mU_* \rho \mU_*^\dagger \mV \rho \mV^\dagger  ) \right] =\frac{ m^{-1} \left| \Tr(\mV^\dagger \mU_*) \right|^2 + 1}{m+1}.
\end{equation}
\end{lemma}
\begin{proof}
WLOG, we rewrite the equation above in terms of a distribution over pure states. Since mixed states are themselves distributions over pure states, this can always be done. Therefore, with a slight abuse of notation, we let $\mathcal{D}$ also denote a distribution over unitary matrices $\mU$ that forms a $2$-design:
\begin{equation}
\begin{split}
    \mathbb{E}_{\rho \sim \mathcal{D}}\left[ \Tr( \mU_* \rho \mU_*^\dagger \mV \rho \mV^\dagger  ) \right] &= \mathbb{E}_{\mU \sim \mathcal{D}}\left[ \bra{0} \mU^\dagger \mV^\dagger \mU_* \mU \ket{0}  \bra{0} \mU^\dagger \mU_*^\dagger \mV \mU \ket{0}  \right] \\
    &= \mathbb{E}_{\mU \sim \mathcal{D}}\left[ \left| \bra{0} \mU^\dagger \mV^\dagger \mU_* \mU \ket{0} \right|^2 \right].
\end{split}    
\end{equation}

Using \eqref{eq:max_entangle_state} and \eqref{eq:max_entangle_state_swap}, we have
\begin{equation}
\begin{split}
    \mathbb{E}_{\mU \sim \mathcal{D}}\left[ \left| \bra{0} \mU^\dagger \mV^\dagger \mU_* \mU \ket{0} \right|^2 \right] &=  \mathbb{E}_{\mU \sim \mathcal{D}}\left[ \bra{0} \mU^\dagger \mV^\dagger \mU_* \mU \ket{0} \bra{0} \mU^\dagger \mU_*^\dagger \mV \mU \ket{0} \right] \\
    &=\mathbb{E}_{\mU \sim \mathcal{D}} \left[  \bra{I_m^2}\bigl( (\mV^\dagger \mU_*) \otimes (\mU_*^\dagger \mV) \otimes \mI  \otimes \mI \bigr) \bigl(\mU \otimes \mU \otimes \bar{\mU} \otimes \bar{\mU} \bigr) \ket{0}^{\otimes 4} \right].
\end{split}    
\end{equation}

Applying \eqref{eq:haar_moments}, we have
\begin{equation}
\begin{split}
    \mathbb{E}_{\mU \sim \mathcal{D}}\left[ \left| \bra{0} \mU^\dagger \mV^\dagger \mU_* \mU \ket{0} \right|^2 \right] &=  \frac{1}{m^2-1} \bra{I_m^2} \bigl( (\mV^\dagger \mU_*) \otimes (\mU_*^\dagger \mV) \otimes \mI  \otimes \mI \bigr) \biggl( \ket{I_m^2}\bra{I_m^2} + \ket{S_m^2}\bra{S_m^2} \biggr) \ket{0}^{\otimes 4} \\
    & \;\;\;\;- \frac{1}{m(m^2-1)} \bra{I_m^2} \bigl( (\mV^\dagger \mU_*) \otimes (\mU_*^\dagger \mV) \otimes \mI  \otimes \mI \bigr) \biggl( \ket{I_m^2}\bra{S_m^2} + \ket{S_m^2}\bra{I_m^2} \biggr) \ket{0}^{\otimes 4} \\
    &= \frac{1}{m^2-1} \left( \Tr[\mV^\dagger \mU_*]\Tr[\mU_*^\dagger \mV] + \Tr[\mV^\dagger \mU_* \mU_*^\dagger \mV] \right)  \\
    &\;\;\;\;- \frac{1}{m(m^2-1)} \left( \Tr[\mV^\dagger \mU_*]\Tr[\mU_*^\dagger \mV] + \Tr[\mV^\dagger \mU_* \mU_*^\dagger \mV] \right) \\
    &= \left(\frac{1}{m^2-1} - \frac{1}{m(m^2-1)} \right) \left( \left| \Tr[\mV^\dagger \mU_* ] \right|^2 + m \right) \\
    &= \frac{m^{-1} \left| \Tr[\mV^\dagger \mU_* ] \right|^2 + 1}{m+1}.
\end{split}    
\end{equation}

\end{proof}

\section{Shallow VQAs as Random Fields}

\subsection{Random Fields on Manifolds}\label{app:rf_review}

Hardness results from barren plateaus or SQ models both intuitively arise from the exponential decay of quantities necessary to perform optimization. To analyze the shallow circuit setting beyond the SQ model---where such exponentially decaying quantities tend not to exist---we look toward models of variational loss landscapes as random fields on manifolds. This mirrors \cite{choromanska2015loss,chaudhari2017energy,anschuetz2022critical} in studying the loss landscapes of machine learning models via mapping to certain random fields which are easier to study analytically. As in \cite{anschuetz2022critical}, here we show that certain classes of variational loss functions of shallow quantum models converge in some limit to \emph{Wishart hypertoroidal random fields} (WHRFs), for which results on the loss landscape are already known~\cite{anschuetz2022critical}. We give a brief review here.

WHRFs in $q$ variables are random fields on a specific tensor product embedding of the hypertorus $\left(S^1\right)^{\times q}$ in $\mathbb{R}^{2^q}$. More specifically, points on this embedding are described by the Kronecker product:
\begin{equation}
    \bm{w}=\bigotimes\limits_{i=1}^q\begin{pmatrix}
    \cos\left(\theta_i\right)\\
    \sin\left(\theta_i\right)
    \end{pmatrix}
\end{equation}
for angles $-\cpi\leq\theta_i<\cpi$. These random fields are then of the form:
\begin{equation}
    F_{\text{WHRF}}\left(\bm{\theta}\right)=\bm{w}^\intercal\cdot\bm{J}\cdot\bm{w},
\end{equation}
where $\bm{J}$ is drawn from the normalized complex Wishart distribution $\mathcal{CW}_{2^q}\left(m,\bm{\varSigma}\right)$ with $m$ degrees of freedom. The complex Wishart distribution is a natural multivariate generalization of the gamma distribution, and is given by the distribution of the square of a complex Gaussian random matrix. Specifically, for $\bm{X}\in\mathbb{C}^{N\times m}$ a matrix with i.i.d. complex Gaussian columns with covariance matrix $\bm{\varSigma}$, the matrix
\begin{equation}
    \bm{W}=\frac{1}{m}\bm{X}\cdot\bm{X}^\dagger
\end{equation}
is normalized complex Wishart distributed with scale matrix $\bm{\varSigma}$ and $m$ degrees of freedom. As discussed in \cite{anschuetz2022critical}, the loss landscapes of WHRFs exhibit a complexity phase transition governed by the \emph{overparameterization ratio}
\begin{equation}
    \gamma=\frac{q}{2m},
\end{equation}
where models with $\gamma\geq 1$ have local minima near the global minimum, and models with $\gamma\ll 1$ have local minima far from the global minimum. Thus, the degrees of freedom parameter $m$ plays a pivotal role in governing the loss landscapes of WHRFs: when $q$ is much smaller than $m$, training is typically infeasible due to an abundance of ``traps'' in the training landscape. Our main result in Section~\ref{app:vqas_as_whrfs} is in demonstrating that even for certain shallow VQAs, the corresponding WHRF is such that $\gamma\ll 1$, and training is infeasible.

\subsection{Shallow VQAs Converge in Distribution to WHRFs}\label{app:vqas_as_whrfs}

As discussed informally in the main text, our goal is to demonstrate that certain distributions of shallow variational quantum algorithms (VQAs) weakly converge to Wishart hypertoroidal random fields (WHRFs). The distribution of local minima of WHRFs was shown in \cite{anschuetz2022critical} to exhibit a phase transition in trainability, where \emph{underparameterized models} are untrainable due to poor local minima, and \emph{overparameterized models} exhibit local minima close to the global minimum (though may still be untrainable for other reasons, e.g. due to barren plateaus~\cite{mcclean2018barren,cerezo2020costfunctiondependent,napp2022quantifying}).

Unlike the nonlocal ansatz case~\cite{anschuetz2022critical}, here we are unable to show the full convergence in distribution of shallow local VQAs to WHRFs. Instead, we focus on the joint distribution of the loss function, gradient norm, and Hessian determinant, where the gradient and Hessian have been normalized by the number of parameters $q$ in the reverse light cone of each term in the Pauli expansion of the problem Hamiltonian; by the parameter shift rule~\cite{schuld2019evaluating,PhysRevA.103.012405}, it is easy to see that this bounds the gradient norm and Hessian eigenvalues as $q$ is large. The local minima results of \cite{anschuetz2022critical} depend only on this joint distribution, and thus showing this convergence suffices for our purposes.

We now review the setup of the VQA loss functions we are considering. Throughout the course of this review, we will make various assumptions, particularly on the distribution of gates in the VQA ansatz and on the independence of various reverse light cones; we discuss these assumptions and whether or not they are reasonable in more detail at the end of this Section. As mentioned in the main text, we consider optimizing VQAs on the problem Hamiltonian $\bm{H}\neq\bm{0}$, which has Pauli decomposition:
\begin{equation}
    \bm{H}=\sum\limits_{i=1}^A\alpha_i\bm{P_i}.
\end{equation}
WLOG, we assume here $\bm{H}$ is traceless, and that all $\alpha_i>0$. To simplify our analysis, we will consider the case where the reverse light cone of each term $\alpha_i\bm{P_i}$ in the Pauli decomposition of $\bm{H}$ is i.i.d. drawn from the same distribution of ansatzes, with the same parameter dependence. To make this more concrete, assume that the reverse light cone of each $\alpha_i\bm{P_i}$ is of the form $\bm{V_i}\left(\bm{\theta}\right)\ket{\bm{0}}$ where $\bm{\theta}\in\mathbb{R}^q$, and has support on a number $l\ll n$ of qubits. In this regime, we can scale and shift the loss landscape of the standard variational loss function
\begin{equation}
    F\left(\bm{\theta}\right)=\bra{\bm{\theta}}\bm{H}\ket{\bm{\theta}}
\end{equation}
to be of the form:
\begin{equation}
    F_{\text{VQE}}\left(\bm{\theta}\right)=1-\lambda_0^{-1}\sum\limits_{i=1}^A\alpha_i\bra{\bm{0}}\bm{V_i}\left(\bm{\theta}\right)^\dagger\bm{P_i}\bm{V_i}\left(\bm{\theta}\right)\ket{\bm{0}}=1+\left\lVert\bm{\alpha}\right\rVert_1^{-1}\sum\limits_{i=1}^A\alpha_i\bra{\bm{0}}\bm{V_i}^\dagger\left(\bm{\theta}\right)\bm{P_i}\bm{V_i}\left(\bm{\theta}\right)\ket{\bm{0}},
    \label{eq:practical_loss_func}
\end{equation}
where $\lambda_0$ is the ground state energy of $\bm{H}$ and $\bm{\alpha}$ is the vector of all $\alpha_i$.

We assume that $\bm{V_i}\left(\bm{\theta}\right)$ is of the form:
\begin{equation}
    \bm{V_i}\left(\bm{\theta}\right)=\bm{W_i}\left(\bm{\theta}\right)\bm{U_i},
    \label{eq:v_def}
\end{equation}
where $\bm{U_i}$ are i.i.d. drawn from an $\epsilon$-approximate $t$-design under the monomial measure on $l$ qubits~\cite{harrow2018approximate}, where $\epsilon=\operatorname{O}\left(1\right)$. Note that in particular, though the total ansatz size $n$ may be large, all potential scrambling of the ansatz may only happen locally, in regions of size $l\ll n$; in other words, these ansatzes are \emph{not} expected to suffer from barren plateaus, particularly if $l=\operatorname{O}\left(\log\left(n\right)\right)$~\cite{cerezo2020costfunctiondependent,napp2022quantifying}. $\bm{W_i}$
is composed of fixed parameterized rotations which we take WLOG to be of the form $\bm{R}_{\bm{Y_a}}\left(\theta_b\right)=\exp\left(-\ci\theta_b\bm{Y_a}\right)$ (where as previously mentioned, this parameter dependence is identical across all $\bm{W_i}$), fixed gates, and potentially randomly chosen gates such that $\bm{W_i}$ itself is a random field. For simplicity, we also assume that all $\theta_i$ are independent from one another (i.e. we are in the $r=1$ regime of \cite{anschuetz2022critical}), and that each qubit in the reverse light cone has at least one parameterized gate. We also assume that the field $\bm{W_i}\left(\bm{\theta}\right)$ is rotationally invariant in $\theta_i$.

We now give the formal statement and proof of the loss landscapes of local, shallow VQAs. First, the formal statement:
\begin{theorem}[Approximately locally scrambled variational loss functions converge to WHRFs]
    Let $p_{\text{VQE},\bm{\theta}}$ be the joint distribution of the loss function of \eqref{eq:practical_loss_func}, its gradient norm, and the determinant of its Hessian at $\bm{\theta}$, where the gradient and Hessian are normalized by $q$. Let $p_{\text{WHRF},\bm{\theta}}$ be the same for the WHRF:
    \begin{equation}
        F_{\text{WHRF}}\left(\bm{\theta}\right)=m^{-1}\sum\limits_{i,j=1}^{2^l}w_i J_{i,j}w_j
    \end{equation}
    with $m=\frac{\left\lVert\bm{\alpha}\right\rVert_1^2}{\left\lVert\bm{\alpha}\right\rVert_2^2}2^{l-1}$ degrees of freedom, where $\bm{J}\sim\mathcal{CW}_{2^l}\left(m,\bm{I_{2^l}}\right)$. Here, $\bm{w}$ are points on the hypertorus $\left(S^1\right)^{\times l}$ parameterized by $\bm{\tilde{\theta}}$, where $\tilde{\theta_i}$ is the sum of all $\theta_j$ on qubit $i$. We then have that $p_{\text{VQE},\bm{\theta}}$ weakly converges to $p_{\text{WHRF},\bm{\theta}}$, up to an error $\operatorname{\tilde{O}}\left(\operatorname{poly}\left(\frac{1}{t}+\epsilon+\exp\left(-l\right)\right)\right)$ in L\'{e}vy--Prokhorov distance.
    \label{thm:app_scramb_vqe_convs_whrfs}
\end{theorem}
As we previously mentioned, for technical reasons, we only prove the convergence of the joint distribution of the loss and certain functions of its first two derivatives. We emphasize once more that this does not affect our final conclusions, as all results on the local minima distribution of WHRFs given in \cite{anschuetz2022critical} depend only on this joint distribution.

To prove Theorem~\ref{thm:app_scramb_vqe_convs_whrfs}, we begin by showing that, up to terms that go to zero polynomially quickly as $\epsilon\to 0,t\to\infty$, one can WLOG consider ansatzes of the form of \eqref{eq:practical_loss_func} that are explicitly Haar random within each reverse light cone of size $l$.
\begin{lemma}[Approximate local scrambling bound on the loss function and its derivatives]
    Let $p_{\text{VQE},\bm{\theta}}$ be the joint distribution described in Theorem~\ref{thm:app_scramb_vqe_convs_whrfs}. Let $p_{\text{Haar},\bm{\theta}}$ be the same, for $\bm{U_i}$ taken to be i.i.d. Haar random. We then have that $p_{\text{VQE},\bm{\theta}}$ weakly converges to $p_{\text{Haar},\bm{\theta}}$, up to an error $\operatorname{\tilde{O}}\left(\operatorname{poly}\left(\frac{1}{t}+\epsilon\right)\right)$ in L\'{e}vy--Prokhorov distance.
    \label{lemma:app_scramb_bound}
\end{lemma}
\begin{proof}
    Let $\phi_{\text{VQE}}\left(\bm{x}\mid\bm{\theta}\right)$ be the joint characteristic function of $p_{\text{VQE},\bm{\theta}}$, and similarly $\phi_{\text{Haar}}\left(\bm{x}\mid\bm{\theta}\right)$. Since $\bm{U_i}$ are assumed to be i.i.d. $\epsilon$-approximate $t$-designs under the monomial measure, for any moments $M_{\text{VQE},\bm{\theta}},M_{\text{Haar},\bm{\theta}}$ of degree $s$ of $p_{\text{VQE},\bm{\theta}},p_{\text{Haar},\bm{\theta}}$, respectively, we have that:
    \begin{equation}
        \left\lvert M_{\text{VQE},\bm{\theta}}-M_{\text{Haar},\bm{\theta}}\right\rvert=\operatorname{O}\left(\epsilon\bm{1}\left[s\leq t\right]+\bm{1}\left[s>t\right]\right).
    \end{equation}
    In particular, for all $T$ sublinear in $t$,
    \begin{equation}
        \left\lvert\phi_{\text{VQE}}\left(\bm{x}\mid\bm{\theta}\right)-\phi_{\text{Haar}}\left(\bm{x}\mid\bm{\theta}\right)\right\rvert=\operatorname{O}\left(\epsilon\operatorname{poly}\left(T\right)+\frac{\left(3T\right)^t}{t!}\right)
    \end{equation}
    for all $\bm{x}$ with $\left\lVert\bm{x}\right\rVert_\infty\leq T$. Similar inequalities hold for the partial derivatives of the joint characteristic functions. Therefore, there exists some $T=\operatorname{\Omega}\left(\operatorname{poly}\left(\min\left(t,\frac{1}{\epsilon}\right)\right)\right)$ such that the second bound of Theorem 4 of~\cite{zaitsev1984} (with $m=\log\left(T\right)$) on the L\'{e}vy--Prokhorov distance is $\operatorname{\tilde{O}}\left(\operatorname{poly}\left(\frac{1}{t}+\epsilon\right)\right)$.
\end{proof}

Until now, we have considered ansatzes with generic parameter dependence. We now show that up to terms vanishing exponentially quickly in the reverse light cone size $l$, we can consider a canonical ansatz form WLOG.
\begin{lemma}[Canonical form for Hamiltonian agnostic variational loss functions]
    Let $p_{\text{Haar},\bm{\theta}}$ be the joint distribution described in Lemma~\ref{lemma:app_scramb_bound}. Let $p_{\text{can},\bm{\theta}}$ be the same for the variational loss function
    \begin{equation}
        F_{\text{can}}\left(\bm{\theta}\right)=\left\lVert\bm{\alpha}\right\rVert_1^{-1}\sum\limits_{i=1}^A\alpha_i\bra{\bm{0}}\bm{R}\left(\bm{\theta}\right)^\dagger\bm{U_i}^\dagger\bm{P_i}\bm{U_i}\bm{R}\left(\bm{\theta}\right)\ket{\bm{0}}+1,
    \end{equation}
    where $\bm{R}\left(\bm{\theta}\right)$ is the product of the parameterized rotations of \eqref{eq:v_def}. We then have that $p_{\text{Haar},\bm{\theta}}$ weakly converges to $p_{\text{can},\bm{\theta}}$, up to an error $\operatorname{\tilde{O}}\left(\operatorname{poly}\exp\left(-l\right)\right)$ in L\'{e}vy--Prokhorov distance.
    \label{lemma:can_form}
\end{lemma}
\begin{proof}
    Let us consider (generally mixed) moments involving random variables of the form:
    \begin{equation}
        K_{ij}\left(\bm{\theta_j}\right)=\bra{\bm{0}}\bm{U_i}^\dagger\bm{W_i}\left(\bm{\theta_j}\right)^\dagger\bm{P_i}\bm{W_i}\left(\bm{\theta_j}\right)\bm{U_i}\ket{\bm{0}}-\bra{\bm{0}}\bm{\tilde{U}_{ij}}^\dagger\bm{U_i}^\dagger\bm{W_i}\left(\bm{\theta_j}\right)^\dagger\bm{P_i}\bm{W_i}\left(\bm{\theta_j}\right)\bm{U_i}\bm{\tilde{U}_{ij}}\ket{\bm{0}},
    \end{equation}
    where $\bm{U_i},\bm{\tilde{U}_{ij}}$ are i.i.d. Haar random on $l$ qubits. By the asymptotic free independence of Haar random matrices from constant matrices, and the fact that
    \begin{equation}
        \tr\left(\bm{W_i}\left(\bm{\theta_j}\right)\bm{P_i}\bm{W_i}\left(\bm{\theta_j}\right)^\dagger\right)=\tr\left(\ket{\bm{0}}\bra{\bm{0}}-\bm{\tilde{U}_{ij}}\ket{\bm{0}}\bra{\bm{0}}\bm{\tilde{U}_{ij}}^\dagger\right)=0,
    \end{equation}
    we have that any such moment is on the order of $\operatorname{O}\left(\operatorname{poly}\exp\left(-l\right)\right)$~\cite{voiculescu1991limit}. In particular, it is easy to see that up to an error in L\'{e}vy--Prokhorov distance on this order, one can WLOG take $p_{\text{can},\bm{\theta}}$ as if the gradient and Hessian components had i.i.d. $\bm{U_{ij}}$ rather than $\bm{U_i}$---for instance, this follows identically to the proof of Lemma~\ref{lemma:app_scramb_bound} with $\epsilon=\operatorname{O}\left(\operatorname{poly}\exp\left(-l\right)\right)$. The result then follows from the unitary invariance of the Haar measure.
\end{proof}

We are now able to prove Theorem~\ref{thm:app_scramb_vqe_convs_whrfs}, following essentially the same procedure as proving Theorem 5 of \cite{anschuetz2022critical}.
\begin{proof}
    By Lemmas~\ref{lemma:app_scramb_bound} and~\ref{lemma:can_form}, $p_{\text{VQE},\bm{\theta}}$ weakly converges to $p_{\text{Haar},\bm{\theta}}$ up to an error $\operatorname{\tilde{O}}\left(\operatorname{poly}\left(\frac{1}{t}+\epsilon+\exp\left(-l\right)\right)\right)$ in L\'{e}vy--Prokhorov distance. By Corollary 1 of \cite{jiang2005}, this then proves weak convergence of $p_{\text{VQE},\bm{\theta}}$ to the corresponding joint distribution of a weighted sum of WHRFs each with $2^{l-1}$ degrees of freedom, up to an additional error in L\'{e}vy--Prokhorov distance exponentially small in $l$. Weak convergence to $p_{\text{WHRF},\bm{\theta}}$ then follows from a trivial generalization of Theorem 5 of \cite{anschuetz2022critical}.
\end{proof}

\paragraph{Scope of results} We now comment on the applicability of the results of \cite{anschuetz2022critical} on the local minima distribution of WHRFs when Theorem~\ref{thm:app_scramb_vqe_convs_whrfs} holds. All analysis of the local minima distribution of WHRFs in \cite{anschuetz2022critical} depends only on the joint distribution $p_{\text{WHRF},\bm{\theta}}$, up to a change in normalization of the gradient and Hessian by $l$ rather than $q$ that does not contribute to the logarithmic asymptotics (i.e. Theorem 7 of \cite{anschuetz2022critical}) when $q\log\left(q\right)=\operatorname{o}\left(m\right)$. Thus, in the discussion of the main text, we take this as an extra assumption. Furthermore, we note that the analysis in the main text holds only up to shifts on the order of $\operatorname{\tilde{O}}\left(\operatorname{poly}\left(\frac{1}{t}+\epsilon+\exp\left(-l\right)\right)\right)$ in the joint distribution $p_{\text{WHRF},\bm{\theta}}$, due to the rate of convergence of Theorem~\ref{thm:app_scramb_vqe_convs_whrfs}. However, shifts on this order do not affect the conclusions of \cite{anschuetz2022critical} for sufficiently large constant $\epsilon^{-1},t$. \change{For completeness, we summarize this discussion and known results on the loss landscapes of WHRFs~\cite{anschuetz2022critical} with the following Corollary:
\begin{corollary}[Shallow, local VQAs have poor loss landscapes]
    Let $F_{\text{VQE}}$ be a local VQA loss function of the form of \eqref{eq:practical_loss_func}. Assume all coefficients $\alpha_i$ of the Pauli decomposition of $\bm{H}$ are $\operatorname{\Theta}\left(1\right)$, and
    \begin{equation}
        l\log\left(n\right)+q\log\left(q\right)=\operatorname{o}\left(2^l A\right).\label{eq:cor_cond}
    \end{equation}
    Then $p_{\text{VQE},\bm{\theta}}$ weakly converges to $p_{\text{WHRF},\bm{\theta}}$ as in Theorem~\ref{thm:app_scramb_vqe_convs_whrfs}, where the associated WHRF has a fraction superpolynomially small in $n$ of local minima within any constant additive error of the ground state energy.
\end{corollary}
\begin{proof}
    The result follows immediately by applying Theorem 7 of \cite{anschuetz2022critical} to Theorem~\ref{thm:app_scramb_vqe_convs_whrfs}.
\end{proof}
}

\paragraph{Assumptions} Let us now discuss in more detail the assumptions made in the course of proving Theorem~\ref{thm:app_scramb_vqe_convs_whrfs}. First, we assume that at least some part of the ansatz circuit scrambles some local region around any measured observable; that is, we assume that the ansatz locally is an $\epsilon$-approximate $t$ design for sufficiently large $\epsilon^{-1},t$. It is known that shallow, local circuits dimensions exhibit this property, when $2$-local Haar random gates are applied~\cite{harrow2018approximate}; thus, in a practical sense, our results assume that the local gates in any distribution of ansatzes under consideration are approximately Haar random. This is a typical model of \emph{Hamiltonian agnostic ansatzes}, where the ansatz is chosen independently from the problem Hamiltonian $\bm{H}$; see for instance the discussion in the main text and the references therein. The inapplicability of this assumption to \emph{Hamiltonian informed ansatzes}---particularly for highly symmetric problems---is discussed in more detail in \change{the Discussion of the main text}, where we review models that may not suffer from the poor trainability properties we show here.

Our other major assumption is the independence of the $\bm{V_i}\left(\bm{\theta}\right)$ (up to the repeated use of parameters). Of course, in practice this is almost never true, as otherwise variational optimization would proceed via optimizing each reverse light cone independently. However, given a problem Hamiltonian $\bm{H}$ and a shallow ansatz, one can consider a subset of Pauli operators in the Pauli decomposition of $\bm{H}$ such that their reverse light cones do not overlap. There is little reason to believe that the loss landscape of this simplified problem should be any more difficult to optimize over than the full problem. We therefore suspect that this assumption is little more than a technical requirement. A similar generalization one could consider is taking the parameters of each $\bm{V_i}$ to being almost entirely independent of one another (though not entirely independent, as one could then optimize each subproblem independently, and $n$ would no longer an accurate measure of the size of the problem). However, in this regime we expect the ``effective'' overparameterization ratio $\gamma$ to go as $\left(\frac{l}{2^l}\right)^A$, as the problem essentially reduces to simultaneously optimizing $A$ loss functions. For $A\sim n$, for instance, this decays exponentially in $n$, and thus we believe that models of this form are also not trainable.

\section{Additional Numerical Experiments}
\subsection{Teacher Student Learning with Checkerboard Ansatz}

\begin{figure*}[t]
        \centering
        \includegraphics[]{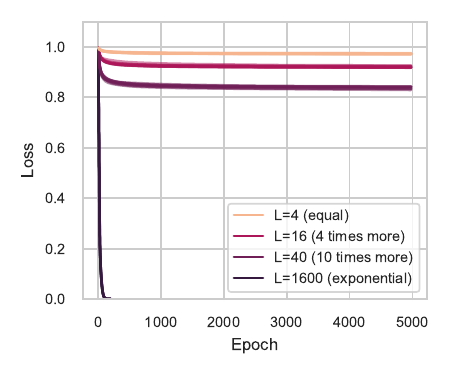}
    \caption{Teacher-student performance evaluation for the depth $L$ checkerboard ansatz. Exponential depth is needed to overparamaterize a model to successfully learn a random circuit of the same form. Here, for each student circuit depth denoted by $L$, $10$ randomly initialized $8$ qubit student circuits are trained to learn a random $L=4$ layer teacher circuit drawn from the same ansatz and parameter distribution.}
        \label{fig:checkerboard}
\end{figure*}

One particular challenge with quantum variational learning is that an overparameterized model needs more parameters than the dimension of the quantum input state (exponential in the number of qubits), whereas classically, overparameterization with respect to the size of the data set typically suffices \cite{choromanska2015loss,li2018learning,arora2019exact}. To illustrate this phenomenon, we consider learning states generated by random shallow checkerboard circuits (denoted the teacher circuit) using checkerboard circuits of the same or more depth (denoted the student circuit). The data set used to train the circuit consist of 512 pairs of inputs randomly drawn from computational basis states with their corresponding output state taken from applying the input state to the teacher circuit. We use the loss $\ell(\ket{\psi},\ket{\phi}) = 1-| \braket{\psi}{\phi}|^2$ to measure the success of learning. Note that, though this is a global loss metric, gradients are analytically calculated to precision sufficient enough to obtain accurate values of the gradients for the relatively small number of qubits considered here.

As shown in \Cref{fig:checkerboard}, exponential depth (and number of parameters) is needed to always successfully learn the data generated by a shallow checkerboard circuit of $4$ layers. We considered ansatzes only over $8$ qubits, which is small enough to be able to feasibly overparameterize the models in our simulations. For fewer qubits and shallower circuits, we found that learning with equal numbers of qubits and layers was sometimes successful; but unsurprisingly, as we show in the main text, learning becomes much harder as qubits are added.

\subsection{Random VQE model}

\begin{figure*}[t]
\centering
    \includegraphics[]{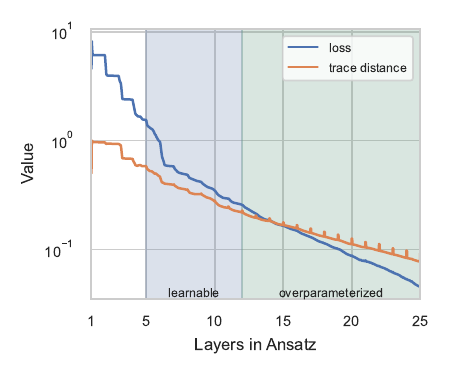}
\caption{When optimizing in a layer-wise fashion, the VQE algorithm converges to a local minimum at each layer until the overparameterized regime where the loss function steadily decreases regardless of the number of layers. Even in the learnable regime where the checkerboard ansatz is capable of expressing the global minima, the ansatz is still unable to find the correct parameters for this global minimum. Bumps in the loss function appear due to small instabilities in training immediately after adding a layer.}
 \label{fig:VQE_final_state_adapt}
\end{figure*}

Here, we empirically analyze the performance of a layer-wise optimizer trained on the random VQE task in the main text. In \Cref{fig:VQE_final_state_adapt}, we train an $11$ qubit ansatz using a layer-wise optimizer \cite{skolik2021layerwise,grimsley2019adaptive}, which initially trains a single layer of the ansatz and adds layers after every 5000 steps to continually add expressiveness. The target Hamiltonian $\mH_t$ here has $4$ layers of perturbations applied to it. Although layer-wise optimizers can avoid issues with barren plateaus \cite{skolik2021layerwise}, our numerical findings clearly show that this does not guarantee the algorithm will avoid traps in the landscape. After $5$ layers, the ansatz has enough parameters to capably express the global optimum (denoted by the label ``learnable''), but nevertheless stalls in optimizing to the ground state. Not until there are at least 12 layers, enough to overparameterize the ansatz with respect to the Hilbert space dimension, does learning smoothly converge to the globally optimal solution.

\subsection{XYZ Hamiltonian Model}
\label{app:XYZ_hamiltonian_section}
\change{
All of the numerical experiments performed elsewhere studied settings where the optimization was performed to numerical precision, two qubit gates were fully parameterized, and the existence of a global minimum at zero loss was guaranteed. The analysis there focused on answering the question of whether convergence to the global minimum is empirically likely to be observed. To study a more realistic setting where such favorable conditions cannot be guaranteed but there still exists hope of some good convergence properties, we turn now to the problem of trying to variationally obtain the ground state of an approximately translationally invariant Heisenberg XYZ Hamiltonian~\cite{heisenberg1928}. Similar Hamiltonians have been studied and analyzed in previous works related to VQE \cite{PhysRevA.92.042303,cade2020strategies}. We perform experiments both with and without Gaussian noise added to the gradients to account for shot noise on a quantum computer. 
}

\change{
The particular target Hamiltonian we aim to optimize is one where qubits are placed on a 2-dimensional grid and interaction terms take place between neighboring qubits. The Hamiltonian takes the form
\begin{equation}
    \mH = \sum_i \mZ_i + \sum_{\langle i,j\rangle} \alpha_{ij} \mZ_i \otimes \mZ_j + \sum_{\langle i,j\rangle} \beta_{ij} \left( \mX_i \otimes \mX_j + 0.66 \mY_i \otimes \mY_j\right),
\end{equation}
where $\langle i,j\rangle$ sums over the neighboring qubits $i$ and $j$ in the grid and $\alpha_{ij}$ and $\beta_{ij}$ are random numbers drawn from the normal distribution with standard deviations set to $0.25$ and means set to $1$ and $3$, respectively. 
}
\begin{table}[ht]
    \centering
\begin{tabular}{llrrrcrrr}
\toprule
   & {} & \multicolumn{3}{l}{energy error} & & \multicolumn{3}{l}{trace distance} \\
   & grid size & $3 \times 2$ & $5 \times 2$ & $7 \times 2$ & &  $3 \times 2$ & $5 \times 2$ & $7 \times 2$ \\
layers & shots &              &              &              &                &              &              \\
\midrule
3  & 10000 &        0.459 &        0.512 &        0.504 &   &       0.983 &        0.999 &        1.000 \\
   & 400 &        0.456 &        0.501 &        0.392 &     &     0.983 &        0.999 &        1.000 \\
   & inf &        0.466 &        0.512 &        0.395 &      &    0.981 &        0.999 &        1.000 \\
9  & 10000 &        0.269 &        0.358 &        0.351 &     &     0.750 &        0.965 &        0.998 \\
   & 400 &        0.343 &        0.434 &        0.386 &        &  0.845 &        0.991 &        0.994 \\
   & inf &        0.245 &        0.350 &        0.344 &         & 0.659 &        0.924 &        0.993 \\
15 & 10000 &        0.104 &        0.293 &        0.303 &        &  0.428 &        0.894 &        0.997 \\
   & 400 &        0.180 &        0.356 &        0.318 &     &     0.577 &        0.965 &        0.987 \\
   & inf &        0.054 &        0.244 &        0.251 &      &    0.293 &        0.842 &        0.968 \\
21 & 10000 &        0.008 &        0.201 &        0.214 &     &     0.162 &        0.799 &        0.982 \\
   & 400 &        0.043 &        0.277 &        0.247 &        &  0.269 &        0.882 &        0.984 \\
   & inf &        0.011 &        0.178 &        0.162 &     &     0.151 &        0.747 &        0.933 \\
27 & 10000 &        0.009 &        0.177 &        0.200 &    &      0.152 &        0.752 &        0.976 \\
   & 400 &        0.034 &        0.254 &        0.251 &       &   0.244 &        0.882 &        0.976 \\
   & inf &        0.010 &        0.122 &        0.129 &        &  0.136 &        0.663 &        0.948 \\
\bottomrule
\end{tabular}
\caption{\change{Error in energy from the ground state (normalized by the magnitude of the ground state energy), and trace distance from the ground state, of a VQE optimizing the Heisenberg XYZ model. Results are averaged across 12 random initializations of the experiment for each entry in the table. Note the poor performance of VQE, particularly at the larger problem sizes.}}
\label{tab:XYZ_hamiltonian}
\end{table}

\change{
As shown in \Cref{tab:XYZ_hamiltonian}, finding the ground state of the XYZ hamiltonian is in general challenging using the ansatz considered. For few layers, the ansatz is not expressible enough to find the target and converges to a poor critical point. For many layers, the VQE algorithm tends to converge to a better optimum, but issues with barren plateaus can begin to arise as indicated by the comparison in performance with assuming infinite shots vs. finite shots. 
}

\section{Details of Numerical Experiments}
\label{app:numerical_experiment_details}

\begin{table}
    \begin{center}
        \small
    \begin{tabular}[t]{@{}llllll@{}}
    \toprule
    Ansatz       & Experiment         & \# Parameters & Optimizer  & Learning Rate                                                                                    \\ \midrule
    QCNN         & Teacher-Student    & $16\ceil{\log_2 n}=\operatorname{O}(\log n)$             & Adam       & 0.001                                                                                            \\
    Checkerboard & Teacher-Student    & $32L\floor{\frac{n}{2}} = \operatorname{O}(nL)$            & Adam       & \begin{tabular}[t]{@{}l@{}}0.001 (underparameterized)\\ 0.0001 (overparameterized)\end{tabular} \\
                 & Random VQE (GD)           &   $128\floor{\frac{n}{2}} = \operatorname{O}(n)$  & vanilla GD & 0.01                                                                                             \\
                 & Random VQE (Adam) & $128\floor{\frac{n}{2}} = \operatorname{O}(n)$  & Adam       & 0.003                                                                                            \\
                 & Adaptive VQE       & $160L = \operatorname{O}(L)$  & Adam       & 0.002 (5\% reduction each layer)     \\
    XYZ ansatz & XYZ Hamiltonian VQE & $7\floor{\frac{L}{3}} = \operatorname{O}(L)$ & Adam & 0.007 (halved every 1000 steps)      
    \\ \bottomrule
    \end{tabular}
    \caption{List of parameter counts, optimizers, and learning rates for the various ansatzes and experiments. $L$ denotes the number of layers and $n$ the number of qubits. }
    \label{tab:hyperparameters}
    \end{center}
\end{table}

All experiments were performed in Python using the Pytorch \cite{NEURIPS2019_9015} package to perform automatic differentiation. Computation was performed on Nvidia RTX\texttrademark$ $ A6000 GPUs. Important hyperparameters for the experiments are listed in \Cref{tab:hyperparameters}. Unless otherwise stated, all gradients were calculated using analytic formulas for automatic differentation with computer precision (32 bit floating point). Therefore, issues with decaying gradients and barren plateaus do not appear in these simulations for the relatively small number of qubits considered. Gradient based optimization was performed using vanilla gradient descent or the Adam optimizer \cite{kingma2014adam}, a popular and effective algorithm for training deep neural networks. We tested other optimizers as well and found no noticeable difference in performance.

\paragraph{Loss surface plot} To generate this plot, we chart the loss landscape at initialization of training in the teacher-student setup of the main text for the $14$ qubit QCNN circuit. The teacher and student circuit were both initialized as described in \Cref{app:qcnn_experimental_backup_details}. 

The loss is plotted along two normalized directions of the parameter landscape. Normalization is applied individually to the $3$ filters of the $14$ qubit QCNN. We loosely follow the ``filter-wise" normalization strategy of \cite{li2018visualizing}, where we first generate a random direction by drawing a value for each parameter from an i.i.d. standard normal distribution. Then, we divide values for the parameters in a given layer by the Frobenius norm of the matrix for the corresponding layer. 

\begin{figure*}[!ht]
 \captionsetup[subfigure]{aboveskip=-1pt,belowskip=-3pt}  
 \begin{subfigure}{0.5\textwidth}
    \includegraphics[trim={2.0cm 1.5cm 2.0cm 2.0cm},clip,width=\linewidth]{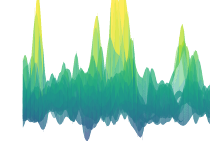}
    \label{fig:surface_backup_3d}
  \end{subfigure}%
  \hspace*{\fill}   
 \begin{subfigure}{0.5\textwidth}
    \includegraphics[width=\linewidth]{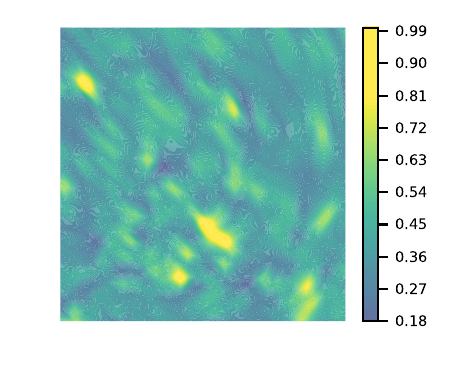}
    \label{fig:surface_backup_contour}
  \end{subfigure}%
  \hspace*{\fill}   
\caption{Loss landscape of the QCNN experiment replicated from the main text, except the initialization of the student circuit is randomly chosen. Here, the global minimum is likely far away and the landscape also appears ``bumpy''; all local minima in the region considered here are far from the global optimum.} \label{fig:surface_backup_both}
\end{figure*}

In the loss surface plot of the main text, we plot the mean squared error loss for the teacher-student task for a batch size of 128 randomly chosen computational basis states. The legend in the plot is shown relative to the maximum value of the loss in the range considered. A value of $0$ here corresponds to the loss at the global minimum. The middle of the plot corresponds to the exact parameters of the teacher circuit, and hence, is a global minimum. This setting is, in a sense, an optimistic setting since initialization is near a global minimum. For comparison, we include in \Cref{fig:surface_backup_both} an example of a loss surface where the student circuit is not initialized near the parameters of the teacher circuit. As is evident in this setting, no longer is there a global minimum in the parameter region considered, and the landscape also appears to be filled with traps.

\subsection{QCNN Experiments}
\label{app:qcnn_experimental_backup_details}

\begin{figure}[!ht]
    \centering
    \includegraphics[]{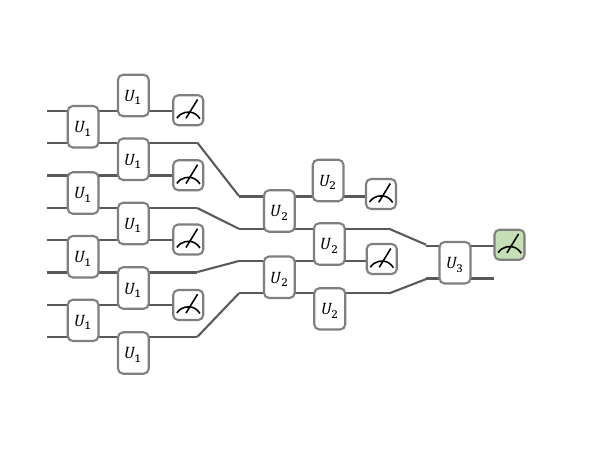}
    \caption{Circuit drawing of QCNN ansatz for 8 qubits. Layers of shared $2$-local unitary transformations are applied followed by measurement of every other qubit. Gates at the edge of the circuit above are applied in a cyclic fashion (i.e. the top and bottom qubit interact). The measurement colored in green is the measurement outcome whose probability we aim to predict in the teacher-student setup. Generically for $n$ qubits, this ansatz has depth $\ceil{\log_2 n}$. During training, the $2$-local unitaries are fully parameterized for our simulations. }
    \label{fig:qcnn_circuit_ansatz}
\end{figure}

The quantum convolutional neural network (QCNN) is an ansatz originally proposed in \cite{cong2019}. This ansatz features parameter sharing across gates in a single layer. The form of this circuit is provided in \Cref{fig:qcnn_circuit_ansatz}. In our experiments, we use the same form of the $2$-local ansatz as in \cite{cong2019} and also studied in \cite{pesah2021absence}. Between convolutional layers, we include no controlled unitary operations based on the measurement outcomes. In learning settings, we fully parameterize the $2$-local unitaries in the skew Hermitian basis of the unitary Lie algebra. To achieve this, we train directly over parameter entries of a matrix $\mM$ and apply $\ce^{\mH}$, where $\mH=\mM - \mM^\dagger$, to perform the resulting unitary transformation. Entries of the matrix $\mM$ were initialized i.i.d. from a standard normal distribution.

For the teacher-student experiments in the main text, we aim to predict the outcome of the final green measurement depicted in \Cref{fig:qcnn_circuit_ansatz} for $512$ randomly chosen computational basis states. For $n$ qubits, the QCNN ansatz for both the teacher and student circuits have $16\ceil{\log_2 n}$ parameters which is a relatively small number compared to the dimension of the Hilbert space. All networks were trained for 5000 epochs and a learning rate of $0.001$ using the Adam optimizer.

\subsection{Checkerboard Ansatz}
\label{app:checkerboard_experimental_backup_details}

\begin{figure}[!ht]
    \centering
    \includegraphics[]{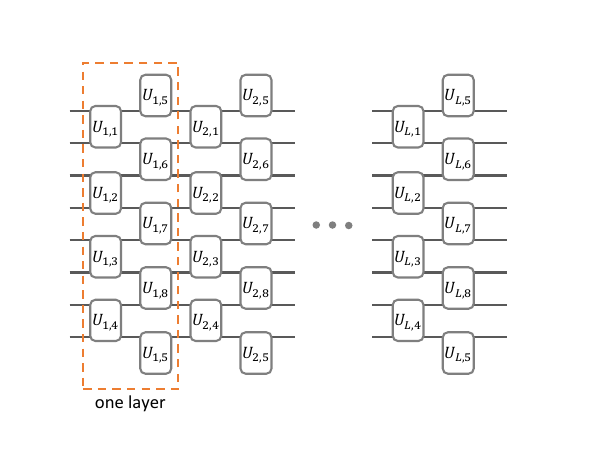}
    \caption{Circuit drawing of the checkerboard ansatz for 8 qubits and $L$ layers. Gates at the edge of the circuit above are applied in a cyclic fashion (i.e. the top and bottom qubit interact). Generically for $n$ qubits, this ansatz has $32L\floor{n/2}$ parameters. During training, the $2$-local unitaries are fully parameterized for our simulations. }
    \label{fig:checkerboard_circuit_ansatz}
\end{figure}

The checkerboard circuit applies gates in a one-dimensional lattice as shown in \Cref{fig:checkerboard_circuit_ansatz}. As in the QCNN experiments, we train directly over parameter entries of a matrix $\mM$ and apply $e^{\mH}$, where $\mH=\mM - \mM^\dagger$, to perform the resulting unitary transformation. Since the exponential map from the Lie algebra is surjective onto the unitary group, this parameterization is capable of expressing any unitary matrix. Entries of the matrix $\mM$ were initialized i.i.d. from a standard normal distribution.

For the teacher-student simulations of the main text, we train networks over 512 randomly chosen computational basis states which is more than the dimension of the Hilbert space and enough information to recover the full unitary transformation. Optimization was performed using the Adam optimizer and a batch size of 128. Networks were trained for 5000 epochs and training was stopped if the loss fell below 0.001 which only occured for the overparameterized setting. We observed that for fewer than 8 qubits, training was successful with very small probability in the underparameterized setting.

\subsection{VQE Experiments on Random Hamiltonians}
\label{app:VQE_experiment_details}
For all of our VQE experiments, the target Hamiltonian $\mH_t$ was constructed by conjugating a local Hamiltonian of $n$ qubits equal to $\sum_{i=1}^n \mZ_i$ with alternating layers of products of two-qubit unitaries $\mU_1$ and $\mU_2$. That is, $\mH_t$ takes the form below as copied from the main text:
\begin{equation}
    \mH_t = \left( \mU_2^\dagger \mU_1^\dagger \right)^L \left[ \sum_{i=1}^n \mZ_i \right] \left( \mU_1 \mU_2 \right)^L + n\mI.
\end{equation}

$\mU_1$ and $\mU_2$ are the tensor product of two-qubit unitaries which for $n$ even take the form:
\begin{equation}
\begin{split}
    \mU_1 &= \mU_1^{(1,2)} \otimes \mU_1^{(3,4)} \otimes \cdots \otimes \mU_1^{(n-1, n)} \\
    \mU_2 &= \mU_2^{(2,3)} \otimes \mU_2^{(4,5)} \otimes \cdots \otimes \mU_2^{(n, n+1)},
\end{split}
\end{equation}
where superscripts above indicate the pair of qubits each $2$-local unitary acts on and indexing is taken $\text{mod } n$. Each $2$ local unitary is drawn from the distribution $e^{\mH}$, where $\mH=\mG - \mG^\dagger$, and each $\mG$ is a $4 \times 4$ matrix with entries drawn i.i.d. from a random normal distribution. Trained unitaries in the checkerboard ansatz are also initialized in this fashion. Optimization is then performed directly on the entries of the matrix in the Lie algebra which form a complete basis for all of the $2$-local unitaries.

\begin{figure*}[!ht]
\centering
    \includegraphics[]{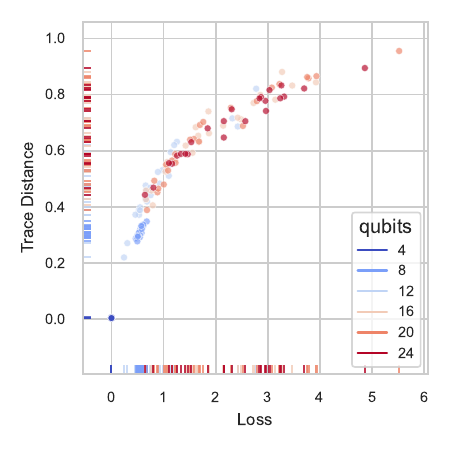}
\caption{Scatter plot showing the values of the loss and trace distance of the final VQE state after 30000 steps of optimization using the Adam optimizer shows that the algorithm converges to poorer local minima as the number of qubits grows. Setting is replicated from the VQE loss experiment from the main text, with the sole change of the optimizer from gradient descent to Adam.} \label{fig:VQE_final_scatter_Adam}
\end{figure*}

In the VQE loss experiment of the main text, each VQE instance was optimized for 30000 steps using a vanilla gradient descent optimizer with a learning rate of 0.01. For completeness, we replicate this plot with the Adam optimizer in \Cref{fig:VQE_final_scatter_Adam} and unsurprisingly observe similar convergence results. All calculations were performed to computer precision, which provides a best-case setting for optimization via real quantum hardware, since gradients and loss function values would have to be calculated using less precise sampling methods on actual quantum computers. In layer-wise VQE experiment of the main text, optimization is performed using an adaptive VQE algorithm similar to the one in \cite{grimsley2019adaptive}. Here, a checkerboard ansatz is initialized as a single layer and optimization is performed layer-wise. We set $n=11$ and small enough such that it is computationally feasible to overparameterize the ansatz. Each 5000 steps of optimization, a layer is added to the ansatz and initialized to the identity mapping. Each additional layer adds $160$ trainable parameters to the ansatz. After each layer is added, the learning rate is multiplied by $0.95$ to make the training more stable with more parameters. At each point in time, all parameters of the ansatz across all layers are trained. For aesthetic purposes and to see the course of training without significant jumps in the plot, we plot a moving average of the values across 10 sequential datapoints in the main text.

\subsection{VQE experiments on XYZ Hamiltonian}
\label{app:XYZ_hamiltonian_ansatz}

\begin{figure}[!ht]
    \centering
    \includegraphics[]{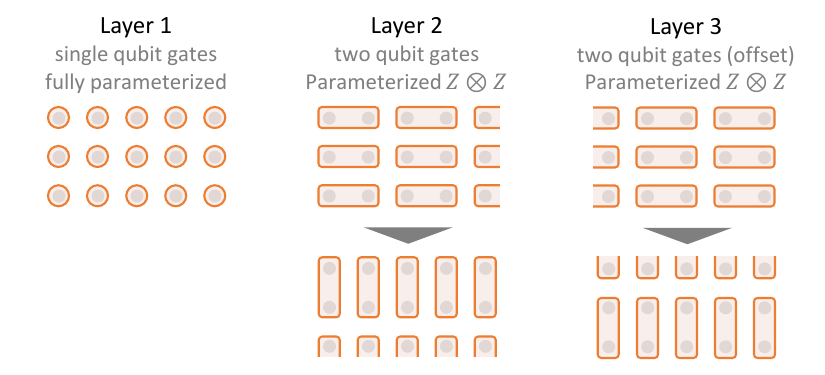}
    \caption{\change{Form of the ansatz used for the XYZ Hamiltonian VQE experiments in \Cref{app:XYZ_hamiltonian_section}. Here, alternating blocks of three layers are composed onto each other. The first layer in each block is a fully parameterized single qubit gate. The next two layers are parameterized Pauli $\mZ \otimes \mZ$ terms to connect all neighboring qubits. Parameters are shared across a layer to better address the near translationally invariance in the model.}}
    \label{fig:XYZ_ansatz}
\end{figure}
\change{
For the XYZ Hamiltonian model empirically analyzed in \Cref{app:XYZ_hamiltonian_section}, we implemented a gate-based ansatz which is fully parameterized in the single qubit gates and parameterized only with Pauli $\mZ \otimes \mZ$ terms for two qubit gates. The form of the ansatz is depicted in \Cref{fig:XYZ_ansatz}. Since the Hamiltonian of the model is approximately translationally invariant in both directions, we implemented sharing of parameters across a layer. Parameters were initialized as random normal variables. Each instance was optimized using the Adam optimizer \cite{kingma2014adam} for 5000 steps. The learning rate was initially set to 0.007 and halved every 1000 steps. For calculations of the trace distance to the ground state, the ground state of the Hamiltonian was calculated by performing an eigendecomposition of the complete Hamiltonian. To account for shot noise, random centered Gaussian noise with standard deviation equal to $1/\sqrt{\text{\# shots}}$ was added to gradients with respect to the parameters.
}

\section{Untrainability Beyond Gradient Descent}\label{app:beyond_gradient_descent}

One may wish to avoid local minima by changing the loss function or performing more advanced versions of gradient-based optimizers. Here, we gives heuristic reasons why these two adjustments will likely not fix any issues of untrainability.

First, we examine changes in the loss function. This is commonly done to avoid barren plateaus and make gradients easier to compute. Let us assume that $\mathcal{L}(\bm{\theta})$ is our original loss function (as a function of the parameters $\bm{\theta}$), which is changed to a new loss function $\tilde{\mathcal{L}}(\bm{\theta})$. Typically, $\tilde{\mathcal{L}}(\bm{\theta})$ is chosen so that it upper and lower bounds $\mathcal{L}(\bm{\theta})$, i.e. $C\tilde{\mathcal{L}}(\bm{\theta}) \leq \mathcal{L}(\bm{\theta}) \leq D\tilde{\mathcal{L}}(\bm{\theta})$ for some constants $C,D$. This guarantees convergence in both metrics when changing the loss function and is the case for e.g. local versions of the inner product and the quantum earth mover's (EM) distance~\cite{9420734,kiani2021}. Now, let us assume that every continuous path from a local minimum at $\bm{\theta_{l}}$ to the global minimum $\bm{\theta^*}$ must increase the loss function by a factor $M>D/C$, i.e. there exists a point in the path that has value at least $M\mathcal{L}(\bm{\theta_l})$. Then, in the new loss function $\tilde{\mathcal{L}}(\bm{\theta_l}) \leq \mathcal{L}(\bm{\theta_l})/C$. Furthermore, at some point in any continuous path, $\mathcal{L}(\bm{\theta})>M\mathcal{L}(\bm{\theta_l})$ which implies that at that point $\tilde{\mathcal{L}}(\bm{\theta}) \geq \mathcal{L}(\bm{\theta})/D>M\mathcal{L}(\bm{\theta_l})/D=\mathcal{L}(\bm{\theta_l})/C$. Thus, $\bm{\theta_l}$ is not within a convex region around the global optimum. This may be too restrictive of an assumption since local minima can often be very shallow, but it also seems to be backed by experiments.

Second, we consider changing the optimization algorithm to a second order optimization algorithm such as in \cite{stokes2020quantum}. These algorithms perform gradient descent by applying a transformation to the gradient of the form:
\begin{equation}
    \bm{\theta_{t+1}} = \bm{\theta_t} - \mu \bm{\varSigma^{+}}\bm{\nabla}_{\bm{\theta_{t}}} \mathcal{L}(\theta_t),
\end{equation}
where $\bm{\varSigma}$ incorporates our second order term, e.g. the Hessian or Fubini--Study metric tensor, and $\bm{\varSigma^{+}}$ is its pseudoinverse. Clearly, in the above, this does not allow one to escape a local minima. Setting $\bm{\theta_t} = \bm{\theta^*}$ above sets the gradient term to zero, and one again one is stuck in a local minimum.

Though other training methods exist, it is not clear \textit{a priori} why they should succeed. For example, training in a layer-wise fashion also does not work as \cite{kiani2021,skolik2021layerwise,campos2021abrupt} show. Finally, note that the above methods can be very effective at alleviating barren plateaus. In fact, changes in metric and second order optimization methods are often precisely designed to fix this issue. Nevertheless, these methods only provably converge to the global optimum in convex or close to convex settings, which is not the case for essentially all variational quantum models.

\bibliography{main}